\documentclass[a4paper]{scrartcl}
\usepackage{microtype}
\usepackage[utf8x]{inputenc}

\usepackage{tikz}

\usepackage{amsthm}
\usepackage{amssymb}
\usepackage{amsmath}
\usepackage{mathrsfs}
\usepackage{xcolor}
\usepackage{enumerate}

\usepackage{ucs,algorithmic,algorithm}

\newtheorem{theorem}{Theorem}
\newtheorem{proposition}[theorem]{Proposition}
\newtheorem{observation}[theorem]{Observation}

\newtheorem{definition}[theorem]{Definition}
\newtheorem{corollary}[theorem]{Corollary}
\newtheorem{lemma}[theorem]{Lemma}
\newtheorem{example}[theorem]{Example}

\newtheorem{remark}{Remark}
\newtheorem{claim}{Claim}

\newcommand{\sg}{\sigma}

\newcommand{\raus}[1]{}
\def\tu#1{\mathbf{#1}}

\newenvironment{frcseries}{\fontfamily{frc}\selectfont}{}

\newcommand{\calG}{{\mathcal{G}}}

\newcommand{\calH}{{\mathcal{H}}}

\newcommand{\calP}{{\mathcal{P}}}

\newcommand{\calR}{{\mathcal{R}}}
\newcommand{\calS}{{\mathcal{S}}}
\newcommand{\calT}{{\mathcal{T}}}
\newcommand{\calU}{{\mathcal{U}}}

\newcommand{\Perm}{\mathrm{Perm}}
\newcommand{\Det}{\mathrm{Det}}
\newcommand{\sgn}{\mathrm{sgn}}

\newcommand{\var}[1]{\textsf{var}({#1})}
\newcommand{\free}[1]{\textsf{free}({#1})}
\newcommand{\atom}[1]{\textsf{atom}({#1})}

\newcommand{\VP}{\mathbf{VP}}

\newcommand{\VNP}{\mathbf{VNP}}
\newcommand{\gapP}{\mathbf{gapP}}
\newcommand{\sP}{\mathbf{\#P}}
\newcommand{\p}{\mathbf{P}}
\newcommand{\NP}{\mathbf{NP}}
\newcommand{\sW}[1]{\mathbf{\#W[#1]}}

\newcommand{\CQ}{\mathrm{CQ}}
\newcommand{\wsCQ}[1]{\mathrm{\#_{#1}CQ}}
\newcommand{\ACQ}{\mathrm{ACQ}}
\newcommand{\CSP}{\mathrm{CSP}}
\newcommand{\sCQ}{\mathrm{\#CQ}}
\newcommand{\sACQ}{\mathrm{\#ACQ}}
\newcommand{\wsACQ}[1]{\mathrm{\#_{#1}ACQ}}

\PrerenderUnicode{é}

\begin{document}

\title{The Complexity of Weighted Counting for Acyclic Conjunctive Queries}
\author{Arnaud Durand\\ 
IMJ UMR 7586  -  Logique\\
Université Paris Diderot\\
F-75205 Paris, France  \\
 {\small \texttt{durand@logique.jussieu.fr}}
%test 
\and
Stefan Mengel\thanks{Partially supported by DFG grants BU 1371/2-2 and BU 1371/3-1.}\\Institute of Mathematics\\ University of Paderborn\\ D-33098 Paderborn, Germany\\ {\small\texttt{smengel@mail.uni-paderborn.de}} 
}

\maketitle

\begin{abstract}
This paper is a study of weighted counting of the solutions of acyclic conjunctive queries ($\ACQ$). The unweighted quantifier free version  of this problem is known to be tractable (for combined complexity), but it is also known that introducing even a single quantified variable makes it $\sP$-hard. We first show that weighted counting for quantifier-free $\ACQ$ is still tractable and that even minimalistic extensions of the problem lead to hard cases. We then introduce a new parameter for quantified queries that permits to isolate large island of tractability. We show that, up to a standard assumption from parameterized complexity, this parameter fully characterizes tractable subclasses for counting weighted solutions of $\ACQ$ queries. Thus we completely determine %This completely delineate 
the tractability frontier for weighted counting for $\ACQ$.
\end{abstract}

\section{Introduction}

Evaluating conjunctive queries is a fundamental problem from database theory. It is equivalent to evaluating so called Select-Project-Join queries and has several equivalent definitions, in particular, in terms of constraint satisfaction problems. While the problem is known to be $\NP$-complete~\cite{CM-77}, a number of structurally restricted classes of conjunctive queries admit efficient algorithms.  Among them,  the class of acyclic conjunctive queries, $\ACQ$ for short, is a large and useful fragment which is well-known to be tractable (see~\cite{Yannakakis-81} but also~\cite{GLS-01}).
The study of acyclic queries has  also been the  starting point of applications of hypergraph decomposition tools to query answering. 

Much less is known about  counting solutions to database queries  which is a basic operation of standard database systems, too, and thus also a very natural fundamental problem. As a generalization of $\mathrm{\#3SAT}$,  it is easy to see that counting solutions to unquantified conjunctive queries is $\sP$-complete. Mostly it has been considered as $\CSP$ with fixed constraint languages without restrictions on the structure of the formulas (see e.g.~\cite{DR-10, BDDJJR-10, DGJ-09}). Recently, Pichler and Skritek \cite{PS-11} showed that the restriction to quantifier free acyclic queries gives tractable instances for counting like it does for decision. This result also extends to other classes for which decision is known to be tractable like bounded hypertree width. Also, Mengel \cite{Mengel-11} showed that weighted counting on bounded treewidth, hence bounded arity, queries is easy, too. So in a sense if we do not allow quantification, not much changes if we go from decision to counting problems.

This picture changes completely if we allow existential quantification of variables. While the complexity of the decision problem remains unchanged, it turns out that counting solutions of general conjunctive queries  becomes much harder since it is complete for $\mathbf{\#\cdot\NP}$ \cite{baulandEtAl05}. Also the acyclic case is no longer tractable. Pichler and Skritek \cite{PS-11} showed that introducing one single existential quantifier allows constructing $\sP$-hard instances of very restricted form. It follows easily that counting solutions to unions of a polynomial number of acyclic conjunctive queries is $\sP$-hard, too, while the decision problem is tractable. This shows that counting and decision differ fundamentally for acyclic conjunctive queries and that in order to find islands of tractability for counting we need new concepts that are tailored specifically to counting problems.

The main contribution of this paper is to give a complete picture of  the (combined) complexity of weighted counting problems related to acyclic conjunctive queries.  We consider weighted counting problems associated to instances of quantifier free $\ACQ$
and present algorithms that compute arithmetic circuits which, in turn, can be evaluated to obtain the result of counting problems. Thus we separate structural manipulations of the CQ instance from the counting process itself. This separation is implicit in other papers (see e.g. \cite{FMR-08}) and is made explicit here. In this context, we first prove (Theorem~\ref{thm:constructVP}, generalizing~\cite{PS-11}) that  computing the sum of the weights of solutions\footnote{Provided the evaluation space itself admits efficient algorithms for multiplication and addition of weights
(such as the field of rationals $\mathbb{Q}$)
 } of a (weighted) quantifier free $\ACQ$ can be done in polynomial time. Then,  
considering extensions of the problem above,
  we show that counting the number of solutions of a conjunction or a disjunction of \textit{two} quantifier free $\ACQ$ is $\sP$-complete even for Boolean domain and for fixed arity (see Proposition~\ref{prop:union}).
 This last result and the fact that one quantification is enough to define hard cases (see~\cite{PS-11}) shows that even  ``minimalistic'' extensions of quantifier free $\ACQ$  lead to intractable counting problems and that one cannot hope to get any meaningful tractable class by these means.  

However, our second set of results counterbalances this impression. We show that a large subclass of $\ACQ$ with existentially quantified variables leads to efficiently solvable counting problems. To this aim, we introduce a  (hypergraph) parameter, called \textit{quantified star size}, to measure 
the degree of dispersion of free variables  in  acyclic conjunctive queries (this generalizes the notion of connex-acyclicity of~\cite{BDG-07}). We show (see Theorem~\ref{thm:countingeasy}) that if this parameter is bounded by some constant $k\in \mathbb{N}$, the resulting (weighted) counting problem is solvable in time $n^{O(k)}$ where $n$ is the size of the instance. Furthermore, we show (see Theorem~\ref{thm:starsizeeasy}) that the quantified star size of a formula is efficiently computable making the counting result applicable. We also show that the runtime bound $n^{O(k)}$ can probably not be improved substantially in the sense that the parameterized problem is not fixed parameter tractable under the standard hypothesis $\mathbf{FPT} \ne \sW{1}$ from parameterized complexity~\cite{FlumGrohe04,FlumGrohe06}. Under the same hypothesis, we finally show (Theorem~\ref{thm:starsizenecessary}) that quantified star size is optimal in the sense that any other structural restriction on acyclic conjunctive queries that leads to polynomial time counting of the solutions must imply bounded quantified star size. 

Our study of counting problem is mostly formulated in the slightly more general setting of arithmetic circuit complexity. So a consequence of our results is that they also provide information about 
the structural complexity of families of polynomials computed by arithmetic circuits and in particular on tractability and intractability in this context. As an example, it is an immediate corollary of our results that the polynomial families we consider which are built over bounded quantified star size queries characterize the class $\VP$ that formalizes efficient computation in the Valiant model~\cite{Valiant-79} and which is well known for its apparent lack of natural non-circuit characterizations. Analog completeness results  can also be derived from our hardness results too.

\paragraph*{Structure of the paper}
Necessary preliminaries about arithmetic circuits, conjunctive queries and acyclicity
% and polynomials built from queries 
are given in Section~\ref{sct:defis}. In Section~\ref{sct:basicacyclic} we show that the weighted counting of solutions to acyclic conjunctive queries is easy extending the result of~\cite{PS-11} but also some results in~\cite{Mengel-11}. In Section~\ref{sct:union} we show that unions of solutions of acyclic conjunctive queries are hard to count. Afterwards we turn to quantified star size to give a parameterization of counting solutions to quantified acyclic queries in Section~\ref{sct:starsize}. In passing, we prove that determining the quantified star size of a query (i.e. of the hypergraph associated to the query) can be done in polynomial time. Here we also show hardness in the sense of parameterized complexity and show that quantified star size likely is the most general restriction that leads to polynomial time counting. In the final Section~\ref{sec:arithmeticcircuitcomplexity} we apply our results to arithmetic circuit complexity: We extend Pichler and Skritek's hardness for quantified conjunctive query to the Valiant model and give a characterization of $\VP$ by acyclic conjunctive queries of bounded star size.

\section{Preliminaries}\label{sct:defis}

\paragraph*{Arithmetic circuit complexity}
%We briefly recall the relevant definitions from arithmetic circuit complexity. A more thorough introduction into arithmetic circuit classes can be found in the book by Bürgisser~\cite{Bur00}. Newer insights into the nature of~$\VP$ are presented in the excellent paper of Malod and Portier~\cite{MP-08}.
An {\em arithmetic circuit} over a field $\mathbb{F}$ is a labeled directed acyclic graph (DAG) consisting of vertices or gates with indegree or fanin $0$ or $2$. The gates with fanin $0$ are called input gates and are labeled with constants from $\mathbb{F}$ or variables $X_1, X_2, \ldots, X_n$. The gates with fanin $2$ are called computation gates and are labeled with $\times$ or $+$. 
The polynomial computed by an arithmetic circuit is defined in the obvious way: An input gate computes the value of its label, a computation gate computes the product or the sum of its childrens' values, respectively. We assume that a circuit has only one sink which we call output gate. We say that the polynomial computed by the circuit is the polynomial computed by the output gate. The \emph{size} of an arithmetic circuit is the number of gates. The \emph{depth} of a circuit is the length of the longest path from an input gate to the output gate in the circuit. A circuit is called \emph{multiplicatively disjoint} if,  for each $\times$-gate, its two input   subcircuits are  disjoint.

\paragraph*{Conjunctive query: decision, counting and weighted counting.} We assume the reader to be familiar with the basics of (first order) logic  (see ~\cite{Libkin-04}). If $\phi$ is a first order formula, $\var{\phi}$ denotes the set of its variables, by $\free{\phi}\subseteq  \var{\phi}$ the set of its free variables and $\atom{\phi}$ the set of its atomic formulas. Let $\tu x=x_1,...,x_k$, we denote $\phi(\tu x)$ the formula with free variables  $\tu x$.
The \textit{Boolean query problem} $\Phi=(\calS,\phi)$ associated to a formula $\phi(\tu x)$  and a structure $\calS$, asks whether the set

\[\phi(\calS)=\{\tu a : (\calS, \tu a)\models \phi(\tu x)\}\]

\noindent called the \textit{query result} is empty or not. The (general) query  problem consists of computing the set $\phi(\calS)$, while the corresponding counting problem is computing the size of $\phi(\calS)$, denoted by $|\phi(\calS)|$. When $\phi$ is a $\{\exists,\wedge\}$-first order formula the boolean query problem is known as the \textit{Conjunctive Query Problem}, $\CQ$ for short.  It is well known that the Boolean $\CQ$ problem is equivalent to the propositional satisfiability problem and thus is $\NP$-complete.
We denote by $\sCQ$  the associated counting problem: given a query instance $\Phi=(\calS,\phi)$, return the value of $|\phi(\calS)|$.

Let now $\mathbb{F}$ be a field and $\calS$ be a finite structure of domain $D$. A $\mathbb{F}$-weight function for $\calS$ is a mapping $w:D\rightarrow \mathbb{F}$. If $\tu a$ is a tuple of elements of $D$ of length $k$, the weight of $\tu a$ is

\[w(\tu a)=\prod_{i=1}^k w(a_i).\]

The \textit{weighted} counting problem for $\CQ$, denoted $\wsCQ{\mathbb{F}}$, is the following problem: given $\Phi=(\calS,\phi)$ and a $\mathbb{F}$-weighted function $w$, return the sum of the weights of all solutions i.e. the value of 

\[\sum_{\tu a \in \phi(\calS)} w(\tu a).\]

When $w$ is the constant function $1$, this value is clearly equal to $|\phi(\calS)|$.

\medskip

\noindent\textsf{Query size and Model of computation.}  The size $\|\Phi\|$ of a query input is the sum of the size of the formula $|\phi|$ and of the size $|\calS|$  of $\calS$. We consider the size $|\calS|$ of $\calS$ with domain $D$ to be the number  of elements in $D$ plus the number of tuples in relations of $\calS$.

All algorithms below are expressed in terms of operations on tuples (without consideration on their size). Hence, we choose the RAM model (with addition as basic operation) under uniform cost measure  as underlying model of computation. However, the choice of a model will have some importance only in the few cases where precise polynomial time bounds are given.

\paragraph*{Acyclic Conjunctive Queries.} 
 A (finite) hypergraph  $\calH$ is a pair $(V,E)$ where $V$ is a finite set and $E\subseteq \calP(V)$. We associate a hypergraph $\calH=(V,E)$ to a formula $\phi$ (the \textit{canonical} structure describing $\phi$) by setting $V:= \var{\phi}$ and $E:= \{\var{a}\mid a \in \atom{\phi} \}$.
 
\begin{definition} A join tree (or tree structure) of a hypergraph $\calH=(V,E)$ is a pair $(\calT,\lambda)$ where $\calT=(V_T,T)$ is a tree and $\lambda$ is a function from $V_T$ to $E$ such that:

\begin{itemize}
\item for each $e\in E$, there is a $t\in V_T$ such that $\lambda(t)=e$,

\item For each $v\in V$, the set $\{t\in V_T: \ v \in \lambda(t)\}$ is a connected subtree of $T$.
\end{itemize}
\end{definition}

A hypergraph is acyclic if it has a join tree~\cite{Fagin-83}. 
When there is no ambiguity, we often identify vertices of a join tree and their labellings. We also sometimes only specify the tree $T$ explicitly, without talking of $\lambda$ which is always implicitly understood to exist.
A formula $\phi$ is acyclic if its associated hypergraph is acyclic.
Considering acyclic $\{\exists,\wedge\}$-first order formulas yields the Boolean acyclic conjunctive query problem denoted $\ACQ$. We denote by $\sACQ$ (resp. $\wsACQ{\mathbb{F}}$) the associated counting (resp. $\mathbb{F}$-weighted counting) problem. If $\phi$ is such that  $\free{\phi}=\var{\phi}$) then $\phi$ is said to be \textit{quantifier-free}.

Any $\tu a\in \phi(\calS)$  will be alternatively  seen as an assignment $\tu a: \var{\phi}\rightarrow D$ or as a tuple of dimension $|\var{\phi}|$. Two arbitrary assignments $\tu a$ and $\tu a'$ are \textit{compatible} if they agree on their common variables.
We will make use of the following classical join operators.

%Let $t\in V_T$, $x\in V$ with $x\in B(t)$ and $t_1,...,t_k$ in $V_T$ the children of $t$ such that $x\in B({t_i})$, $i\leq k$. We define a choice function $\frcc$, which given $t\in V_T$, $x\in V$ with $x\in B(t)$, choose a $t_i$, $i\leq k$, such that $x\in B({t_i})$.

\begin{definition} Let $\phi(\tu x,\tu y)$, $\psi(\tu y,\tu z)$ be two conjunctive queries with  $\tu x\cap\tu z = \emptyset$  and let $\calR,\calS$ be two finite structures. Let us define:

- The natural join:  $\phi(\calR)\bowtie \psi(\calS) = \{(\tu a,\tu b,\tu c): \ (\tu a,\tu b)\in \phi(\calR)  \mbox{ and } (\tu b,\tu c)\in \psi(\calS)\}$

- The special (left)  join: $\phi(\calR)\ltimes \psi(\calS) = \{(\tu a,\tu b): \ (\tu a,\tu b)\in \phi(\calR)  \mbox{ and there exists } (\tu b,\tu c)\in \psi(\calS)\}$.
\end{definition}

When $\calR=\calS$, $\phi(\calR)\bowtie \psi(\calS)$ is simply $[\phi\wedge\psi](\calR)$.
The natural join between two relations can be computed in time linear in the size of the relations (here in time linear in $|\phi(\calR)|$ and $|\psi(\calS)|$).

\begin{remark}[CQ and CSP] A well-known equivalent formulation of the quantifier-free conjunctive query problem can be stated in terms of $\CSP$. In this later problem, given two structures $\calS$ and $\calT$, one asks whether there exists an homomorphism from $\calS$ to $\calT$. In the recent past, counting and weighted counting for $\CSP$ have been mainly stated in the \textit{non uniform} version of the problem i.e. when the template $\calT$ is fixed and only $\calS$ is given as input (see e.g.~\cite{DR-10, BDDJJR-10, DGJ-09}). In this paper, we focus on the uniform version of the problem.   
\end{remark}

\paragraph*{Polynomials defined by conjunctive queries}\label{intro:polynomials}

We briefly introduce a polynomial $Q$ that generalize $\sCQ$. A more thorough discussion of $Q$ can be found in~\cite{Mengel-11}.

Let $\Phi=(\phi, \calS)$ be a query with domain $D$. We assign to $\Phi$ the following polynomial $Q(\Phi)$ in the variables $\{X_d \mid d \in D\}$.

\[Q(\Phi) := \sum_{a\in \phi(\calS)} \prod_{x \in \var{\phi}} X_{a(x)} = \sum_{a\in  \phi(\calS)} \prod_{d \in D} X_d^{\mu_d(a)},\]
where $\mu_d(a) = |\{ x \in \var{\phi} \mid a(x) = d\}|$ computes number of variables mapped to $d$ by $a$.
Note that the number of variables in $Q(\Phi)$ is $|D|$, the size of the domain. 
and that  $Q(\Phi)$ is homogeneous of degree $|\var{\Phi}|$.

Observe that $Q$ is essentially the weighted counting problem $\wsCQ{\mathbb{F}}$ by setting $X_d:= w(d)$ for all $d\in D$. Thus if we can efficiently compute circuits that in turn compute $Q(\Phi)$, we can efficiently solve $\wsACQ{\mathbb{F}}$ on the instance $\Phi$.

\paragraph*{Parameterized counting complexity}

This section is a very short introduction to parameterized counting complexity (for more details see \cite{FlumGrohe04,FlumGrohe06}). 

A parameterized counting problem is a function $F: \Sigma^* \times \mathbb{N} \rightarrow \mathbb{N}$, for an alphabet $\Sigma$. Let $(x,k)\in \Sigma^*\times \mathbb{N}$, then we call $x$ the input of $F$ and $k$ the parameter. A parameterized counting problem $F$ is fixed parameter tractable, or $F\in \mathbf{FPT}$, if there is an algorithm computing $F(x,k)$ in time $f(k)\cdot |x|^c$ for some computable function $f:\mathbb{N}\rightarrow \mathbb{N}$ and some constant $c\in \mathbb{N}$.

Let $F:\Sigma^*\times \mathbb{N} \rightarrow \mathbb{N}$ and $G:\Pi^*\times \mathbb{N} \rightarrow \mathbb{N}$ be two parameterized counting problems. A parameterized parsimonious reduction from $F$ to $G$ is an algorithm that computes for every instance $(x,k)$ of $F$ an instance $(y,l)$ of $G$ in time $f(k)\cdot |x|^c$ such that $l\le g(k)$ and $F(x,k) = G(y,l)$ for computable functions $f,g:\mathbb{N}\rightarrow \mathbb{N}$ and a constant $c\in \mathbb{N}$. A parameterized $T$-reduction from $F$ to $G$ is an algorithm with an oracle for $G$ that solves any instance $(x,k)$ of $F$ in time $f(k)\cdot |x|^c$ in such a way that for all oracle queries the instances $(y,l)$ satisfy $l\le g(k)$ for computable functions $f,g$ and a constant $c\in \mathbb{N}$. 

Let $p$-$\mathrm{\#Clique}$ be the problem of counting $k$-cliques in a graph where $k$ is the parameter and the graph is the input. A parameterized problem $F$ is in $\sW{1}$ if there is a parameterized parsimonious reduction from $F$ to $p$-$\mathrm{\#Clique}$\footnote{Let us remark that Thurley \cite{Thurley-06} gives good arguments for defining $\sW{1}$ not with parsimonious reductions. He instead defines $\sW{1}$ with parameterized $T$-reductions with only one oracle call. We keep the definition of \cite{FlumGrohe04,FlumGrohe06}, because we will show no $\sW{1}$ upper bounds and thus can avoid these subtleties. We remark though that finding the right reduction notions for counting problems is notoriously tricky to get right (see e.g. \cite{KPZ99,DHK05}).}. $F$ is $\sW{1}$-hard, if there is a parameterized $T$-reduction from $p$-$\mathrm{\#Clique}$ to $F$. As usual, $F$ is $\sW{1}$-complete if it is in $\sW{1}$ and hard for it, too.

A standard assumption from parameterized complexity is that not all problems in $\sW{1}$ (and thus in particular the complete problems) are fixed parameter tractable. Thus, from showing that a problem $F$ is $\sW{1}$-hard it follows that $F$ can be assumed to be not fixed parameter tractable.

Except for these definitions we will not use parameterized parsimonious reductions and we will not use the complete power of parameterized $T$-reductions either. Instead, all parameterized reductions in the remainder of the paper will be $T$-reductions with exactly one oracle call.

\section{Constructing circuits for acyclic conjunctive queries}\label{sct:basicacyclic}

\begin{theorem}\label{thm:constructVP}
 Given an acyclic quantifier free conjunctive query $\Phi$, we can in time polynomial in $\|\Phi\|$ compute a multiplicatively disjoint arithmetic circuit $C$ that computes~$Q(\Phi)$.
\end{theorem}

\begin{proof} The first step follows that of \cite{Mengel-11}. Since arity of queries is not bounded, a new approach is necessary to show that only polynomially many gates are necessary to compute $Q(\Phi)$. Also the algorithmic nature of the construction is stressed more to give the upper bound on the complexity of constructing the circuit $C$.

So let $\Phi=(\calS,\phi)$ be an acyclic conjunctive query. Let $(\calT, \lambda)$ the join tree associated with $\phi$. By definition, the tree $\calT$ has $m$ vertices $t_1,...,t_m$ associated to the atoms $\lambda(t_1), ..., \lambda(t_m)$ in $\phi$. Observe that $\calT$ can be constructed from $\phi$ in polynomial time; indeed it can even be computed in logarithmic space, see \cite{GLS-01}. Thus we do not consider the construction of $\calT$ but take it as given.
For $t\in V_T$, we call  $\phi_t$ the conjunction of constraints corresponding to the subtree $\calT_t$ with $t$ as root. The set $\var{\phi_t}=\bigcup_{t'\in \calT_t} \var{\lambda(t')}$ is denoted by $e_t$. For convenience we also denote the atomic formula $\lambda(t)$ by $\lambda_t$.

 Let $\tu a$ be an assignment of some variables of $\phi$ and $c\subseteq \var{\phi_t}$. 
We show by induction on the depth of $\calT$ that the following polynomial can be computed by an arithmetic circuit of polynomial size
\[
f_{t,\tu a, c} = \sum_{\substack{\tu \alpha \in \phi_t(\calS)\\ \tu \alpha \sim \tu a}} \prod_{x\in  c} X_{\tu a(x)}.
\]

Remark that if $r$ is the root of $\calT$ then,  $f_{r, \emptyset, \var{\phi}}=Q(\Phi)$. Observe that in contrast to \cite{Mengel-11} we have an exponential number of polynomials $f_{t,\tu a, c}$, so we cannot afford to compute them all in a bottom up fashion. Instead we will construct the circuit top down starting from $r$ and make sure that in each step only polynomially many $f_{t,\tu a, c}$ are needed. This will directly give the runtime bound for the construction.

So suppose first that $t$ is a leaf.
Then, 
$\phi_t$ is some atomic constraint $\lambda_t$, hence
$\phi_t(\calS)$  is of size at most linear in $|\calS|\leq \|\Phi\|$, so the sum only involves a number of terms linear in~$\|\Phi\|$.

 Suppose now $t\in V_T$ is not a leaf and let $t_1,...,t_k$ in $V_T$ be the children of $t$ in $\calT$.
Let $c_0,c_1,...,c_k$ be a partition of $c$ into disjoint sets such that each $c_i\subseteq e_i\cap c$, for $i=1,...,k$ and $c_0\subseteq c\backslash \bigcup_{i=1}^k e_{t_i}$ (the need for choosing a partition is that each variable $x$ appearing in several  $\lambda(t_i)$ sets must be taken into account at most once in order not to overcount the exponent of $X_{a(x)}$). 

\[
\begin{array}{rl}
f_{t,\tu a, c} =  & \displaystyle \sum_{\substack{\tu \alpha \in \phi_t(\calS)\\ \tu \alpha \sim \tu a}} \prod_{x\in c} X_{\tu a(x)}\\
  = &    
  \displaystyle \sum_{\substack{\tu \alpha \in \phi_t(\calS)\\ \tu \alpha \sim \tu a}}
  \prod_{x\in c_1} X_{\tu \alpha(x)}
  \cdots
  \prod_{x\in c_k} X_{\tu \alpha(x)}
   \prod_{x\in c_0} X_{\tu \alpha(x)}
    \end{array}
\]

Let $A_t=((\lambda_t(\calS)\ltimes \phi_{t_1}(\calS)) \ltimes \phi_{t_2}(\calS)) \ltimes \ldots \ltimes \phi_{t_k}(\calS)$. Note that $A_t\subseteq \lambda_t(\calS)$.

\begin{claim}\label{claim} The set $A_t$ is computable in time $|\calS|\times |\phi_t|$.
\end{claim}

\begin{proof}[Proof of the claim] By induction on the tree depth. The proof is an adaptation of Yannakakis algorithm to evaluate acyclic conjunctive queries (see \cite{Yannakakis-81}).
Note that $A_t$ is a subrelation of $\phi_t(\calS)$ with   $\phi_t$ being a constraint of the input formula. So if $t$ is a leaf, the result is obvious.
If $t$ is not a leaf. Let $t_1,...,t_k$ be its children. Remark that
\[A_t = (\lambda_t(\calS)\ltimes A_{t_1}) \ltimes A_{t_2}) \ltimes \ldots \ltimes A_{t_k},\]
\noindent since each $A_{t_i}$ is the projection of $\phi_{t_i}(\calS)$ onto $\var{\phi_{t_i}}$. 
The set $A_t$ is computed as follows: Special joins are computed step by step  respecting the order given by parentheses. To compute each join sort the two sets of assignments using the lexicographic ordering 
induced by the variables they have in common. Then, run once through the two sorted relations to select the right tuples.
Suppose now that each $A_{t_i}$ is computable in time $ |\calS|\cdot |\phi_{t_i}|$. Then, $A_t$ is computable in time $|\calS|\cdot (|\lambda_t|+ |\phi_{t_1}| + \ldots + |\phi_{t_k}|)\leq |\calS|\cdot|\phi_t|$.
\end{proof}

Each solution $\tu \alpha\in \phi_t(\calS)$ can be uniquely expressed as the natural join of a tuple $\tu\beta\in A_t$ and a sequence of $\tu \alpha_i\in \phi_{t_i}(\calS)$, $i=1,...,k$, compatible with $\beta$ (more formally by natural join of singleton relations containing  these tuples), i.e.
given $\tu \alpha\in \phi_t(\calS)$, there exist $\tu\beta\in A_t$ and  $\tu \alpha_i\in \phi_{t_i}(\calS)$, $i=1,...,k$, such that

\[\{\tu\alpha\} = \{\tu\beta\} \bowtie \{\tu \alpha_1\} \bowtie \ldots \bowtie \tu \{\alpha_k\}.\]
 
 Conversely, given $\tu\beta\in A_t$ and a sequence of $\tu \alpha_i\in \phi_{t_i}(\calS)$, $i=1,...,k$, compatible with $\beta$, the natural join of these tuples is an $\tu \alpha\in \phi_t(\calS)$. This follows from the connectedness condition in the join tree, i.e. from the fact that given distinct $i,j\leq k$, $\var{\phi_{t_i}}\cap \var{\phi_{t_j}}\subseteq \var{\lambda_{t}}$. Indeed, if $\alpha_i$ and $\alpha_j$ assign values to of a common variable, they must agree on it, because they both agree with $\beta$. This implies that the following equalities hold.

\[
\begin{array}{rl}
f_{t,\tu a, c} =  &   
  \displaystyle \sum_{\substack{\tu \alpha \in \phi_t(\calS)\\ \tu \alpha \sim \tu a}}
  \prod_{x\in c_1} X_{\tu \alpha(x)}
  \cdots
  \prod_{x\in c_k} X_{\tu \alpha(x)}
   \prod_{x\in c_0} X_{\tu \alpha(x)}\\
   = &   
 \displaystyle\sum_{\substack{\tu \beta \in A_t\\ \tu \beta \sim \tu a}}
  \sum_{\substack{\tu \alpha_1 \in \phi_{t_1}(\calS)\\ \tu \alpha_1 \sim \tu \beta}}
  \cdots 
  \sum_{\substack{\tu \alpha_k \in \phi_{t_k}(\calS)\\ \tu \alpha_k \sim \tu \beta}}
  \prod_{x\in c_1} X_{\tu \alpha_1(x)}
  \cdots
  \prod_{x\in c_k} X_{\tu \alpha_k(x)}
   \prod_{x\in c_0} X_{\tu \beta(x)}\\
   = & \displaystyle\sum_{\substack{\tu \beta \in A_t\\ \tu \beta \sim \tu a}}
   f_{t_1,\beta,c_1} \cdots f_{t_k,\beta,c_k} \cdot   \prod_{x\in c_0} X_{\tu \beta(x)} 
    \end{array}
\]

Note that  the sum is now over $A_t$ and not over $\phi_t(\calS)$ anymore.
We claim that the construction described above can be done in polynomial time. Indeed, for each $t$ we only have to compute the $f_{t,\tu a, c}$ for one fixed set $c$ but for potentially all $\tu a \in A_{t'}$ where $t'$ is the father of vertex $t$ in $\calT$. Thus for $t$ we only have to compute $|A_{t'}|\le |\calS|\le \|\Phi\|$ polynomials $f_{t,\tu a, c}$. Furthermore for $t$ we only have to access the polynomials for $t_1,...,t_k$, more precisely all of  $f_{t_1,\beta,c_1}$, \ldots, $f_{t_k,\beta,c_k}$ with $\beta\in A_t$.  So the computation of one $f_{t,\tu a, c}$ involves only $O(|A_t|\times (k+|c_0|))=O(|\calS|\times |\phi|)$ arithmetic operations. Computing the $f_{t,\tu a, c}$ for all $\tu a\in A_{t'}$ but fixed $t$ can then be done with $O(|\phi| \times |\calS|^2)$ operations. Summing up over all $t$ we get a total upper bound of $O(|\phi|^2 \times|\calS|^2)$, so the circuit $C$ for $Q(\Phi)$ is of polynomial size. Now remark that each set $A_t$, for $t\in V_T$ can be constructed in time $O(|\phi| \times |\calS|)$ by Claim~\ref{claim}. For a fixed $\tu a$, filtering all elements $\beta$ of $A_t$ compatible with $\tu a$ can be done in linear time after sorting $A_t$. Hence, the index set of each sum is efficiently computable and the construction of the circuit can be done in polynomial time.

In a final step we apply the construction of Malod and Portier \cite{MP-08} to make the circuit multiplicatively disjoint. 
\end{proof}

We get the following corollary on weighted counting problems (which generalizes  a recent result of Pichler and Skritek \cite{PS-11}).

\begin{corollary} Let $\mathbb{F}$ be field such that iterated addition and multiplication are computable in polynomial time in $\mathbb{F}$. Then 
 $\wsACQ{\mathbb{F}}$ can be solved in polynomial time for quantifier free queries.
\end{corollary}
\begin{proof}
 Given an instance $\Phi$, Theorem \ref{thm:constructVP} yields a circuit $C$ that computes $Q(\Phi)$. Setting $X_d=w(d)$ for all $d\in D$, we can evaluate $Q(\Phi)$ to give the answer to the weighted counting problem. Observe that efficient evaluation is possible, because the degree of $Q(\Phi)$ is bounded and thus we can use standard depth reduction techniques to avoid a blowup of the size of representations of field elements.
\end{proof}

\section{Union and intersection of acyclic queries}\label{sct:union}

In this section, we show that considering conjunction and disjunction of two acyclic conjunctive queries leads to intractable counting problems.

\begin{proposition}\label{prop:union}
 Computing the size of the union and the intersection of solutions to two quantifier free $\sACQ$-instances are both $\sP$-complete. This result remains true for $\sACQ$ on boolean domain and arity at most $3$.
\end{proposition}

\begin{remark}
In~\cite{GSS-01}, it is proved that the (bi-)colored grid homomorphism problem is $\NP$-complete. This result implies part of Proposition~\ref{prop:union}, i.e. that counting the assignments of the conjunction of two $\ACQ$-instances is $\sP$-complete (the fact that this hardness result is still true on Boolean domain does not follow, however).
\end{remark}

For the proof we use the following lemma:

\begin{lemma}\label{lem:grid}
 Counting solutions to quantifier free conjunctive queries whose primal graph is a grid is $\sP$-complete even for domains of size $4$.
\end{lemma}
\begin{proof}
Counting solutions to general quantifier free conjunctive queries is in $\sP$, so we only need to show hardness. We show hardness by reducing a restricted version of $\mathrm{\#circuitSAT}$ to $\sCQ$ with the desired grid structure. From the $\sP$-completeness of our $\mathrm{\#circuitSAT}$ version we get $\sP$-hardness for counting solutions of conjunctive queries with grid structure.

We now define this version of $\mathrm{\#circuitSAT}$ that we call $\#$($\land$-$\neg$-$\mathrm{grid}$)-$\mathrm{circuitSAT}$: An instance of $\#$($\land$-$\neg$-$\mathrm{grid}$)-$\mathrm{circuitSAT}$ is a boolean circuit which only contains $\land$- and $\neg$-gates and in which all gates are vertices of a $2$-dimensional grid. Furthermore, the edges of the circuit are non-intersecting paths along the edges of the grid.

\begin{proposition}\label{prop:gridhard}
$\#$($\land$-$\neg$-$\mathrm{grid}$)-$\mathrm{circuitSAT}$ is $\sP$-complete under parsimonious reductions.
\end{proposition}
\begin{proof}
We make a parsimonious reduction from $\mathrm{\#circuitSAT}$. Let $C$ be a $\mathrm{\#circuitSAT}$ instance, i.e. a boolean circuit. In a first step we substitute all $\lor$-gates $x\lor y$ by $\neg(\neg x \land \neg y)$. We then make sure that every gate has at most degree $3$ and that all input gates and the output gate have at most degree $2$ by adding double negations. Call the resulting circuit~$C'$.

We now embed $C'$ into a grid. To do so we take a three step approach that starts with a coarse grid that is then refined. Let $n$ be the size of $C'$. We first distribute the vertices that represent gates into a $n\times n$-grid $G_1$ such that each vertex of depth $i$ has the coordinates $(i,j)$ for some $j$. Furthermore each edge of the circuit is a sequence of straight lines where each straight line goes from a vertex in one row to another vertex in the next row. Also in each vertex of $G_1$ at most two lines start and end. For vertices on which no gate of $C'$ lies, we assume that at most one edge starts and ends. It is clear that such an embedding can be constructed easily.

In a second step we make sure that the edges of the circuit follow the edges of a grid without congestion. We do this for each row of the coarse grid $G_1$ individually. We construct a new grid $G_2$ by adding $2n-1$ new rows before each row in $G_1$ and one new column before each column. Observe that each vertex $(i,j)$ in $G_1$ has the coordinates $(2ni, 2j)$ in $G_2$. Each vertex $v$ of $G_1$ in row $i$ has a most $2$ outgoing straight lines $l_1, l_2$ representing edges of the circuit $C'$ which both end in a vertex of row $i+1$. Let $l_1$ end in $(i+1, j)$ and $l_2$ end in $(i+1, j')$ with $j<j'$, then we call $l_1$ be the low output and $l_2$ the high output. If there is only one output, we define it to be high. We also make the equivalent definition for high and low inputs. 

Now we substitute the lines representing edges of $C'$ by paths in $G_2$. Let $l$ be a line that starts in $G_1$ in $(i,j')$ and ends in $(i+1, j)$. We construct a path $P_l$ from $(2ni, 2j)$ to $(2n(i+1), 2j')$:
\begin{itemize}
 \item If $l$ is a low output and a low input the path is the piecewise linear curve through the vertices $(2ni, 2j)(2ni, 2j-1)(2ni+ 2j, 2j-1)(2ni+ 2j, 2j')(2n(i+1), 2j')$.
 \item If $l$ is a high output and a low input the path is through $(2ni, 2j)(2ni+2j+1, 2j)(2ni+ 2j+1, 2j')(2n(i+1), 2j')$. 
 \item If $l$ is a low output and a high input the path is through $(2ni, 2j)(2ni, 2j-1)(2ni + 2j, 2j-1)(2ni+ 2j, 2j'+1)(2n(i+1), 2j'+1)(2n(i+1), 2j')$. 
\item If $l$ is a high output and a high input the path is through $(2ni, 2j)(2ni+2j+1, 2j)(2ni+ 2j+1, 2j'+1)(2n(i+1), 2j'+1)(2n(i+1), 2j')$. 
\end{itemize}
The result is an embedding of $C'$ into a grid such that the gates are on vertices of $G_2$ and the edges of $C'$ are paths in the grid. Observe that the paths were constructed in such a way that two paths between gates never share edges, so they only intersect in single vertices.

In the final step of the reduction we get rid of these intersections on non-gate vertices by adding additional gates. Each crossing in $G_2$ is substituted by the gadget illustrated in Figure \ref{fig:gridgadget}. To do so we make the grid finer again by a constant factor. The result is a circuit $C''$ that is embedded into a grid. Furthermore $C''$ has the same satisfying assignments as $C'$.

\begin{figure}[t]
 \centering
\begin{tikzpicture}[scale=.8, transform shape]
% original crossing
\node (a1) at (2,0) {a};
\node (ae1) at (2,4) {};
\node (b1) at (0,2) {b};
\node (be1) at (4,2) {};
% with xor gates
\node (a2) at (6,2) {b};
\node (ae2) at (10,2) {};
\node (b2) at (8,0) {a};
\node (be2) at (8,4) {};
% simulate xor
\node (a3) at (12,2) {x};
\node (ae3) at (15,4) {$x\oplus y$};
\node (b3) at (14,0) {y};
\node (be3) at (14,4) {};
\tikzstyle{every node}=[draw,shape=circle];
% with xor gates
\node (1) at (8,2) {$\oplus$};
\node (2) at (8,3) {$\oplus$};
\node (3) at (9,2) {$\oplus$};
% simulate xor
\node (4) at (13,3) {$\neg$};
\node (5) at (14,3) {$\land$};
\node (6) at (15,3) {$\lor$};
\node (7) at (14,2) {$\lor$};
\node (8) at (15,2) {$\land$};
\node (9) at (15,1) {$\neg$};
% original crossing
\draw [->] (a1) -- (ae1);
\draw [->] (b1) -- (be1);

% with xor gates
\draw [->] (a2) -- (7,2) -- (1);
\draw [->] (7,2) -- (7,3) -- (2);
\draw [->] (b2) -- (1);
\draw [->] (1) -- (2);
\draw [->] (1) -- (3);
\draw [->] (8,1) -- (9,1) -- (3);
\draw [->] (2) -- (be2);
\draw [->] (3) -- (ae2);

% simulate xor gate
\draw [->]  (13,2) -- (4);
\draw [->] (a3) -- (7);
\draw [->]  (14,1) -- (9);
\draw [->] (b3) -- (7);
\draw [->] (4) -- (5);
\draw [->] (7) -- (5);
\draw [->] (5) -- (6);
\draw [->] (7) -- (8);
\draw [->] (9) -- (8);
\draw [->] (8) -- (6);
\draw [->] (6) -- (ae3);

\end{tikzpicture}
 \caption{The crossing paths in the left are substituted by a gadget without crossings in the middle that uses $\oplus$-gates which compute $\mathrm{xor}$ of its inputs. It is easily checked that the outputs compute $(a\oplus b) \oplus a$ and $(a\oplus b) \oplus b$ which simplify to $b$ and $a$ respectively. On the right we show how the $\oplus$-gates can be simulated over the basis $\land, \lor, \neg$ without losing planarity. Degree $4$ gates, splitting of edges and $\lor$-gates can be avoided by introducing some more $\neg$-gates and using De Morgan's law.}
 \label{fig:gridgadget}
\end{figure}
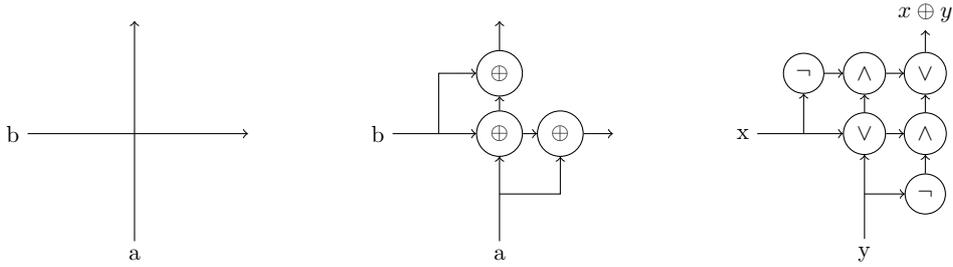

\end{proof}

\begin{remark}
 We could also have given a proof of Proposition \ref{prop:gridhard} with results on embedding general planar graphs into grids in the way we need it (see e.g. \cite{Valiant-81}). We have chosen to present an ad-hoc proof instead to keep the results of this paper self-contained.
\end{remark}

We now reduce $\#$($\land$-$\neg$-$\mathrm{grid}$)-$\mathrm{circuitSAT}$ to $\sCQ$ instances of grid structure. So let $(C,G)$ be an instance of $\#$($\land$-$\neg$-$\mathrm{grid}$)-$\mathrm{circuitSAT}$, i.e. a circuit $C$ that is embedded into a grid $G$. Let $G$ be of size $n\times n$. W.l.o.g. we may assume that no gates are on neighboring vertices in $G$ and that the output gate is not a $\land$-gate. For each $\land$-gate af $C$ we arbitrarily fix one input as the first input while the other one is the second one. We construct a binary conjunctive query $\Phi$ whose primal graph is $G$. The domain is $\{0,1,2,3\}$ where $0$ and $1$ represent the usual boolean values while $2$ and $3$ are used in a gadget construction for $\land$-gates. For each edge $e=uv$ in $G$ we add a constraint $\phi_e$ in the variables $u$ and $v$ in the following way:
\begin{itemize}
  \item If $e$ is not an edge of $C$, $\phi_e$ has the satisfying assignments $\{ab \mid a,b \in \{0,1,2,3\}\}$.
 \item If $e$ is an edge of $C$ directed from $u$ to $v$ and $v$ is not a gate and $u$ is not a $\land$-gate, $\phi_e$ has the satisfying assignments $\{00, 11\}$.
 \item If $e$ is an edge of $C$ directed from $u$ to $v$ and $v$ is a $\neg$-gate, $\phi_e$ has the satisfying assignments $\{01, 10\}$.
\item If $e$ is an edge of $C$ directed from $u$ to $v$ and $v$ is a $\land$-gate and the path to $v$ over $u$ is from the first input of $v$, $\phi_e$ has the satisfying assignment $\{00, 02, 11, 13\}$.
\item If $e$ is an edge of $C$ directed from $u$ to $v$ and $v$ is a $\land$-gate and the path to $v$ over $u$ is from the second input of $v$, $\phi_e$ has the satisfying assignment $\{00, 01, 12, 13\}$.
\item If $e$ is an edge of $C$ directed from $u$ to $v$ and $v$ is not a gate and $u$ is a $\land$-gate, $\phi_e$ has the satisfying assignments $\{00, 10, 20, 31\}$.
\end{itemize}

Observe that the construction near the $\land$-gates is possible, because no two gates are neighbors. So the constraints are all well defined.
Now each vertex that is not part of $C$ gets a unary constraint that has only the single satisfying assignment $1$. Also the output gate of $C$ gets such a unary constraint.

We claim that if we fix an assignment $a$ to the variables representing the inputs of $C$, there is an satisfying extension to the other variables if and only if $a$ satisfies $C$. Furthermore, this extension is unique. It is clear that the constraints along the paths and on the $\neg$-gates propagate the correct values along the grid. In a satisfying assignment, the variable representing an $\land$-gate has to take the value representing the values of its inputs in binary. The gates after the $\land$-gates then calculate the conjunction value for these inputs.
\end{proof}

\begin{proof}[Proof of Proposition \ref{prop:union}]
 Again, we only need to show hardness. By the inclusion-exclusion principle counting for unions and intersections is equally hard, so it suffices to show hardness for intersections. The reduction is straightforward with Lemma \ref{lem:grid}. Let $\Phi$ be a conjunctive query whose primal graph is a grid. We separate the constraints into two new formulas: $\Phi_1$ gets all the constraints that lie on rows of the grid, $\Phi_2$ gets those on the columns. Clearly we have $\Phi = \Phi_1 \land \Phi_2$ and the $\Phi_i$ are acyclic. Thus the first part of the lemma follows.
 
 To show that the result is true for queries on boolean domain, we sketch a different encoding of $\land$-$\neg$-$\mathrm{grid}$-$\mathrm{\#circuitSAT}$ into conjunctive queries. Roughly speaking, the structure of the encoding is basically the same but non boolean elements are mapped to sequences of boolean variables  (that represent their binary encodings). To do so we need ternary relations.  For completeness, details are given below. Again let $G$ be the $n\times n$  grid and suppose no gates are on neighboring vertices in $G$ and that the output gate is not a $\land$-gate.  For each $\land$-gate $v$, we introduce a second vertex/variable $v_1$. We construct a ternary CQ instance $\Phi$ as follows.  For each edge $e=uv$ in $G$ we add an constraint $\phi_e$ in the following way:
\begin{itemize}
  \item If $e$ is not edge of $C$ , $\phi_e$ has the satisfying assignments $\{00, 01, 10, 11\}$.
 \item If $e$ is an edge of $C$ directed from $u$ to $v$ and $v$ is not a gate and $u$ is not a $\land$-gate, $\phi_e$ has the satisfying assignments $\{00, 11\}$.
 \item If $e$ is an edge of $C$ directed from $u$ to $v$ and $v$ is a $\neg$-gate, $\phi_e$ is the constraint  with the following satisfying assignments $\{01, 10\}$.
\item If $e$ is an edge of $C$ directed from $u$ to $v$ and $v$ is a $\land$-gate and the path to $v$ over $u$ is from the first input of $v$, $\phi_e$ is the ternary constraint on variables $u,v,v_1$ with the following satisfying assignment set $\{000, 010, 101, 111\}$.
\item If $e$ is an edge of $C$ directed from $u$ to $v$ and $v$ is a $\land$-gate and the path to $v$ over $u$ is from the second input of $v$, $\phi_e$ is the ternary constraint on variables $u,v,v_1$ which has the satisfying assignment $\{000, 001, 110, 111\}$.
\item If $e$ is an edge of $C$ directed from $u$ to $v$ and $v$ is not a gate and $u$ is a $\land$-gate, $\phi_e$ is the constraint on variables $u_1,u,v$ which has the satisfying assignments $\{000, 010, 100, 111\}$.
\end{itemize}

The constraint is then split into two conjunctive query instances $\Phi_1$ and $\Phi_2$ as above grouping horizontal and vertical constraints separately. Note that, connection at $\and$ gates $v$ between hyperedges is now on two vertices  $v$ and $v_1$. But the resulting hypergraphs for $\Phi_1$ and $\Phi_2$ are well acyclic.
\end{proof}

The reductions of this section are all parsimonious, so we directly get the following corollary~\footnote{We state this corollary for completeness. Although we found no references, it is certainly already known}:

\begin{corollary}
 Deciding if the intersection of the solutions of two quantifier free acyclic conjunctive queries is nonempty is $\NP$-hard.
\end{corollary}

Note that in contrast, it is well-known that deciding the disjunction of acyclic conjunctive queries can be done in time linear in $\|\Phi\|$. 

\section{Quantified star size}\label{sct:starsize}

\subsection{Definitions and statement of the results}\label{sct:results}

It is proved in \cite{PS-11} that introducing one single existential quantifier in 
acyclic conjunctive queries leads to $\sP$-complete problems. So bounding the number of quantified variables does not yield tractable instances.
 In this section, we will show that not the number of quantified variables is crucial but how they are distributed in the associated hypergraph. 
 A basic observation on the hard instance in \cite{PS-11} is that the formula has a star (in the graph theoretical sense) in its associated graph whose center is the quantified variable. 
 Abstracting this observation, we introduce a parameter called \textit{quantified star size} that leads to tractable $\sACQ$ instances.

Before we formulate the main results of this section, we make several definitions. 

Let $\calH= (V,E)$ be a hypergraph and $S \subseteq V$. 
The \emph{induced subhypergraph}  $\calH[S]$ of $\calH$ is the hypergraph $\calH[S] = (S, \{e \cap S \mid e \in E, e \cap S \ne \emptyset\})$. The induced subhypergraph of an edge set $E'\subseteq E$ is $\calH[E'] = (\bigcup_{e\in E'}e, E')$. 
%Let $A\subseteq E$, then $\calH - A := \calH[E\setminus A]$. 
Let $x,y \in V$, a \textit{path} between $x$ and $y$ is a subset of edges $e_1,...,e_k\in E$ such that $x \in e_1$, $y \in e_k$, and for all $i\leq k-1$, $e_i \cap e_{i+1}\neq \emptyset$. 
 Alternatively, a path can be seen as the sequence of vertices $(x,a_1,...,a_k,y)$ such that 
 $a_i\in e_i\cap e_{i+1}$.% When $x=y$ the path is called a \textit{cycle}.
These definitions apply to graphs as well.

\begin{observation}\label{obs:acyclicitytransmission} If $\calH$ is an acyclic hypergraph and $C\subseteq V(H)$, then $\calH[C]$ is acyclic. If $\calT=(V_T,T)$ is a join tree of $\calH$ then $\calT[C]$, the tree obtained from $T$ by replacing for all vertices $t\in V(T)$ the labeling $\lambda(t)$ by $\lambda(t)\cap C$, is a join tree of $\calH[C]$. 
\end{observation}
\begin{proof}
Immediate. $\calT[C]$ is a subforest of $\calT$. The connectedness condition of the set  $\{t\in V_{T[C]}: \ v \in \lambda(t)\}$, for all $v\in C$ is obviously true. 
\end{proof}

\begin{definition}[$S$-component~\cite{BDG-07}]
Let 
$\calH=(V,E)$ be a hypergraph and $S\subseteq V$. 
Let $E_{\nsubseteq S}$ be the set of hyperedges $\{e\in E: e \nsubseteq S\}$. 
The \emph{$S$-component} of $e\in E_{\nsubseteq S}$ is the hypergraph $\calH[E']$ where $E'$ is the set of all edges $e'\in E_{\nsubseteq S}$ such that there is a path from $e-S$ to $e'-S$ in $\calH[V-S]$.
\end{definition}

It is clear that hyperedges of $ E_{\nsubseteq S}$ are partitioned into pairwise disjoint $S$-components.  

\begin{definition}[$S$-$k$-star, $S$-star size] Let $\calH=(V,E)$ be a hypergraph, $S\subseteq V$ and $k\in\mathbb{N}$. The subhypergraph $\calH'=(V',E')$ of $\calH$ is a \emph{$S$-$k$-star} if:
\begin{itemize}
\item $\calH'$ is an $S$-component of $\calH$.
\item there exist $y_1,...,y_k\in V'\cap S$ such that there is no edge $e\in E$ that contains more than one of the $y_i$.
\end{itemize}
We say that $y_1, \ldots , y_k$ form the $S$-$k$-star.

The \emph{$S$-star size} of $\calH$ is the maximum $k$ such that there is a $S$-$k$-star in $\calH$. 
\end{definition}

In other words, the $S$-star size of a hypergraph $\calH$ is the maximal star size of its $S$-components.

\begin{observation}\label{obs:inducedstarsize}
Let $\calH=(V,E)$ be a hypergraph, $S\subseteq V$ and $\calH'$ be an $S$-component of $\calH$. Then, if $\calH$ is acyclic, $\calH'$ is acyclic.
\end{observation}
\begin{proof}
Let $\calT$ be a join tree of $\calH$.
An $S$-component $\calH'=(V',E')$ is a subhypergraph induced by the edge set $E'$. By definition $E'$ is connected in $G$ and it follows that $\{t\mid \lambda(t)\in E'\}$ induces a subtree $\calT'$ of $\calT$. The connectedness condition holds in $\calT$ and thus it holds in $\calT'$, too. It follows that $\calT'$ is a join tree of $\calH'$ and $\calH'$ is acyclic.
\end{proof}

\begin{definition}
 The \emph{quantified star size} of a acyclic conjunctive formula $\phi(\tu x)$ is the $S$-star size of the hypergraph $\calH$ associated to $\phi(\tu x)$, where $S$ is the set of free variables in~$\phi(\tu x)$.
\end{definition}

\begin{example} The formula $\phi(x,y)\equiv \exists t \exists z  R(x,y,t) \wedge S(x,z,t)$ has quantified star size $1$ because the free variables $x$ and $y$ appear together in one constraint.

Paths formulas (of arbitrary length), for example $\phi(x,y,z)\equiv \exists t_1 \exists t_2   R(x,t_1) \wedge R(t_1,z) \wedge R(z,t_2) \wedge R(t_2,y)$, are of quantified star size $2$. 

Star formulas, such as $\phi(x,y,z,t)\equiv \exists u  R(u,x) \wedge R(u,y) \wedge R(u,z) \wedge R(u,t)$ have quantified star size  equal to the degree of the center of the star (here $4$). 
\end{example}

\begin{example}~\label{example:formulahard} The hard formula of \cite{PS-11} is of quantified star size $n$, the size of the structure  domain.
\end{example}

We now formulate the main results of this section. The first result is that bounding the quantified star size yields tractable counting problems.

\begin{theorem}\label{thm:countingeasy}
 There is an algorithm that given an acyclic conjunctive query $\Phi$ computes an arithmetic circuit $C$ that computes $Q(\Phi)$. The runtime of the algorithm is $\|\Phi\|^{O(k)}$ where $k$ is the quantified star size of $\Phi$.
\end{theorem}

\begin{corollary}
 There is an algorithm for the problem $\sACQ$ that runs in time $\|\Phi\|^{O(k)}$ where $k$ is the quantified star size of the input query $\Phi$.
\end{corollary}

The second result below implies  that computing the quantified star size is easy and thus classes of $\sACQ$-instances of bounded quantified star size are efficiently decidable.

\begin{theorem}\label{thm:starsizeeasy}
 There is a polynomial time algorithm that, given a hypergraph $\calH=(V,E)$ and $S\subseteq V$, computes the $S$-star size of $\calH$.
\end{theorem}

We prove Theorem~\ref{thm:countingeasy} and Theorem~\ref{thm:starsizeeasy} in the following two subsections.

\subsection{Computation of $S$-star size}

In this section we show that $S$-star size can be computed in polynomial time.

Let $\calH=(V,E)$ be a hypergraph and $S\subseteq V$. We say that $E^* \subseteq E$ covers $S$ if $S\subseteq \bigcup_{e\in E^*} e$. If $S=V$ we say that $E^*$ is an edge cover of $\calH$. An independent set $I$ in $\calH$ is a set $I\subseteq V$ such that there are no distinct vertices $x,y\in I$ that lie in a common edge $e\in E$.

\begin{lemma}\label{lem:edgecover}
 For acyclic hypergraphs the size of a maximum independent set and a minimum edge cover coincide. Moreover, there is an algorithm that given an acyclic hypergraph $\calH$ computes a maximum independent set $I$ and a minimum edge cover $E^*$ of $\calH$.
\end{lemma}

The first sentence in Lemma \ref{lem:edgecover} can be seen as an adaptation of K\H{o}nig's theorem for bipartite graphs (see e.g.~\cite{Bollobas-98}) to acyclic hypergraphs.
The proof uses a minimally modified version of an algorithm that Guo and Niedermeier \cite{GuoN-06} describe to compute minimum (unweighted) edge covers of acyclic hypergraphs. We show here that their techniques cannot only be used to compute minimum edge covers but also maximum independent sets of acyclic hypergraphs.

\begin{proof}
 Clearly the size of any independent set is not greater than that of any edge cover, simply because no edge can cover two vertices in an independent set. So if we present an algorithm that computes an independent set $I$ and an edge cover $E^*$ of a given acyclic hypergraph $\calH=(V,E)$ such that $|I|=|E^*|$ we are done.

So let us now describe an algorithm that computes $I$ and $E^*$: Let $\calT= (V_T, E_T)$ be a join tree of $\calH$ with root $r$. We start with initially empty sets $I$ and $E^*$ and iteratively delete leaves of $\calT$ in a bottom-up manner from the leaves to the root. It is easily seen that for each leaf $t\in V_T$, either $\lambda(t)\subseteq \lambda(t')$ where $t'$ is the parent of $t$ or there exists $y\in\lambda(t)$ such that $y\not\in\lambda(t')$. In this case, we will say that $y$ is unique for $t$. If $t=r$ is a leaf of $\calT$, i.e. $r$ is the only vertex in $\calT$, we say by convention that if $\lambda(t)$ contains any vertices, they are all unique for $t$.
 
We do the following until $V_T$ is empty. First, choose a leaf $t$ of $\calT$.
 If there is no vertex unique for $t$, we simply delete $t$ from $V_T$.
 If there  are vertices that are unique for $t$, choose one vertex $y$ among them and add it to $I$. Furthermore, add $\lambda(t)$ to $E^*$, delete all vertices in $\lambda(t)$ from $\calH$ and delete $t$ from $V_T$. When $V_T$ is empty, $I$ and $E^*$ are the result of the algorithm.

For a vertex $t\in V_T$ we denote by $\calT_t$ the subtree of $\calT$ with the root $t$. Let furthermore $V_t$ be defined as the vertices in $V$ that appear only in $\{\lambda(t^*) \mid t^*\in V(\calT_t)\} \subseteq E$ and in no other edge in $E$. 

\begin{claim}\label{clm:cover}
 Whenever the algorithm deletes $t\in V_T$, the edge set $E^*$ covers the vertices $V_t$.
\end{claim}
\begin{proof}
Assume that the claim is false, then there is a first vertex $t\in V_T$ met during the execution of the algorithm for which after $t$ is deleted some vertex $y\in V_t$ is not covered by $E^*$. For all children $t^*$ of $t$ the vertices in $V_{t^*}$ are covered by $E^*$, so $y$ must lie in $\lambda(t)$. But then $y$ is unique for $t$ before $t$ is deleted. Thus $\lambda(t)$ is added to $E^*$ and $y$ is covered by $\lambda(t)\in E^*$ after $t$ is deleted which is a contradiction.
\end{proof}

From Claim \ref{clm:cover} it follows directly that $E^*$ is an edge cover at the end of the algorithm.

\begin{claim}\label{clm:independent}
 At each point in time during the algorithm $I$ is an independent set in $G$.
\end{claim}
\begin{proof}
Assume the claim is wrong. Then, there is a first vertex $y$ that is added to $I$ such that $y$ is adjacent to $x$ already in $I$. The vertex $x$ was added to $I$, so there was $t\in V_T$ such that $x$ was unique for $t$ when $t$ was considered by the algorithm. Thus $x$  is in $V_t$ and consequently not in  $\lambda({t'})$ for any vertex $t'\in V_T\setminus V(\calT_t)$. Hence, if $x$ and $y$ are adjacent,  there must be a vertex $t^*\in V(\calT_t)$ such that $\{x,y\}\subseteq \lambda(t^*)$. But $y$ is added to~$I$ after $x$ and thus it must appear in $\lambda({t'})$ for a vertex $t'\in V_T\setminus V(\calT_t)$. Then because of the connectedness condition and the fact that $\calT$ is a tree, $y$ must also be in $\lambda(t)$ and thus is deleted from $\calH$ when $t$ is deleted. But then $y$ cannot be added later which is a contradiction.
\end{proof}

With Claim \ref{clm:cover} and Claim \ref{clm:independent} we have that at the end of the algorithm $E^*$ is an edge cover of $G$ and $I$ is an independent set in $G$. It is easy to see, that $|E^*| = |I|$. This completes the proof.
\end{proof}

\begin{corollary}\label{cor:computestarsize}
 Let $\calH=(V,E)$ be an acyclic hypergraph and $S\subseteq V$. Then the following statements are true:
\begin{enumerate}[a)]
 \item~\label{item:computing star size} The $S$-star size of $\calH$ can be computed in polynomial time.
 \item~\label{item:computing cover} Let $\calH'=(V',E')$ be an $S$-component of $\calH$ and let $k$ be the $S$-star size of $\calH'$. There is a polynomial time algorithm that computes an edge set $E^* \subseteq E'$ that covers $S\cap V'$ and $|E^*|= k$.
\end{enumerate}
\end{corollary}
\begin{proof}
\ref{item:computing star size}) Let $\calH_1,\ldots,\calH_m$ be the $S$-components of $\calH$. By Observation~\ref{obs:inducedstarsize}, each $\calH_i= (V_i, E_i)$, $i\in \{1,\ldots, m\}$, is acyclic and then by Observation~\ref{obs:acyclicitytransmission}, $\calH_i[S]$ is acyclic too. By Lemma~\ref{lem:edgecover}, for each $i\in \{1,\ldots, m\}$, one can determine the size of a maximum independent set $I_i$ of $\calH_i[S]$. But we claim that for each $i$ the star size $s_i$ of the $S$-component $\calH_i$ and the size of the maximum independent set in $\calH_i[S]$ coincide. Indeed, consider two vertices $x,y \in S\cap V_i$ such that there is an edge $e\in E\setminus E_i$ such that $\{x,y\}\in e$. Note that $x$ and $y$ are each included in at least one edge of $E_i$.   Remark also that $(V_i,E_i\cup E[V_i\cap S])=\calH[V_i]$ where $E[V_i\cap S]:= \{e\cap (V_i \cap S) \mid e\in E, e\cap (V_i \cap S)\ne \emptyset\}$. $\calH[V_i]$ is acyclic by Observation~\ref{obs:acyclicitytransmission}. Let $\calT_i'$ be a join tree of $\calH[V_i]$. The vertices $\{t\in V(\calT_i')\mid \lambda(t)\in E_i\}$ are connected in $\calH[V_i]$, so they induce a subtree  $\calT_i$ of $\calT_i'$. 
But then, if $e \in E\backslash E_i$, by the connectedness condition the vertex $t$ with $\lambda(t)=e$ must be connected to $\calT_i$ by two different paths that enter $\calT_i$ via two different edges (since both $x$ and $y$ are in distinct edges of $E_i$). This contradicts the fact that $\calT_i'$ is a tree. Thus $s_i$ is indeed the size of a maximum independent set made of $S$-vertices in $\calH[V_i]$ which is the $S$-star size of $\calH_i$. The $S$-star size of $\calH$ is then the maximal value among  $s_1,...,s_m$ and the result follows.

\ref{item:computing cover})  We compute an edge cover $\tilde{E}$ of size $k$ for $\calH'[S]$ with Lemma \ref{lem:edgecover}. Then for each edge $\tilde{e} \in \tilde{E}$ one can easily find an edge $e\in E'$ with $\tilde{e}\subseteq e$.
\end{proof}

% \begin{corollary}\label{cor:computestarsize}
%  Let $\calH=(V,E)$ be an acyclic hypergraph and $S\subseteq V$. Then the $S$-star size of $\calH$ can be computed in polynomial time.
% \end{corollary}
% \begin{proof}
% Let $\calH_1,\ldots,\calH_m$ be the $S$-components of $\calH$. By Observation~\ref{obs:inducedstarsize}, each $\calH_i$, $i\in \{1,\ldots, m\}$, is acyclic and then by Observation~\ref{obs:acyclicitytransmission}, $\calH_i[S]$ is acyclic too. By Lemma~\ref{lem:edgecover}, for each $i\in \{1,\ldots, m\}$, one can determine the size of a maximum independent set $I_i$ of $\calH_i[S]$, hence the star size $s_i$ of the $S$-component $\calH_i$. The $S$-star size of $\calH$ is then the maximal value among  $s_1,...,s_m$.
% \end{proof}
% 
% With the same argument one sees that one can also compute small covers for $S$ in the $S$-components of an acyclic hypergraph.
% 
% \begin{corollary}\label{cor:computecover}\stefan{should we prove this?}
%  Let $\calH$ be an acyclic hypergraph and $\calH'$ be an $S$-component. Let $k$ be the $S$-star size of $\calH'$. There is a polynomial time algorithm that computes an edge set $E^* \subseteq E(\calH')$ that covers $S\cap V(\calH')$ and $|E^*|= k$.
% \end{corollary}

\subsection{Efficiently computing the $Q$-polynomial}

We now have the necessary ingredients to prove Theorem \ref{thm:countingeasy}.

\begin{proof}[Proof of Theorem \ref{thm:countingeasy}] Let $\Phi=(\calS,\phi)$ be an input query of quantified star size $k$. We will construct a quantifier free formula $\varphi$ and a new structure $\calS'$ in time $\|\Phi\|^{O(k)}$ such that $\phi(\calS)=\varphi(\calS')$.

Let $\calH$ be the hypergraph of $\phi$ and $S$ the set of free variables of $\phi$. 
Let $\calH'=(V',E')$ be an $S$-component of $\calH$ and let $\phi'$ be the subformula of $\phi$ whose  atomic formulas are the hyperedges of $E'$. The formula $\phi$ can then be written as a conjunction $\phi'\wedge \psi$, where the formula $\psi$ contains all the atoms of $\phi$ not in $\phi'$.
By the definition of $S$-components we have $(\var{\phi'}\backslash \free{\phi'})\cap \var{\psi}=\emptyset$. In other words, the quantified variables in $\phi'$ only appear in atoms of $\phi'$ and common variables of $\phi'$ and $\psi$ are necessarily in $S$, i.e.~they are free.

Let now $E^*\subseteq E'$ be a cover of $V'\cap S$ of size $s\leq k$ computed with Corollary~\ref{cor:computestarsize}. Let  $\phi_1, \ldots, \phi_s$ be the atomic formulas associated to edges in $E^*$.  We will compute $\phi'(\calS)$ and construct in parallel a new atomic constraint $\varphi_1$ and a new relation $\varphi_1^{\calS}$ such that $\varphi_1^{\calS}=\phi'(\calS)$. The set of variables $\var{\varphi_1}$ is $\bigcup_{i=1}^s \free{\phi_i}=\free{\phi'}$. For each combination $t_1, \ldots , t_s$ of tuples in $\phi_1(\calS), \ldots, \phi_s(\calS)$ we add the tuple $t:= t_1 \bowtie \ldots \bowtie t_s$ to the relation $\varphi_1^{\calS}$ if
\begin{itemize}
 \item the tuples $t_1, \ldots, t_s$ are consistent, i.e. they coincide on shared variables,
 \item the ACQ instance  that we get from $\phi'$ by fixing the variables in $\var{\varphi_1}$ to the values specified by the tuples $t_1, \ldots, t_s$ is satisfiable.
\end{itemize}

Observe that we have to only consider $\prod_{i=1}^s |\phi_i(\calS)|\le\|\Phi\|^{s} \le \|\Phi\|^k$ combinations $t_1, \ldots , t_s$ and the resulting Boolean queries can each be evaluated in time  $O(|\phi'|\cdot |\calS|)$ by Yannakakis' algorithm (see~\cite{Yannakakis-81}). Thus the construction of $\varphi_1$ and the relation $\varphi_1^{\calS}$ can be done in time $\|\Phi\|^{O(k)}$. Let us call $\calS'$ the union of the structure $\calS$ and  $\varphi_1^{\calS}$. 
We then have a new query $\varphi_1\wedge \psi$ such that $[\varphi_1\wedge \psi](\calS')=\phi(\calS)$. 

Let $\calT$ be a join tree of $\phi$. We can choose a subtree $\calT'$ of $\calT$ such that $\calT'$ is a join tree of $\calH'$ by considering the induces subgraph of the vertices $\{t\mid \exists e\in E', \lambda(t)=e\}$. Also $\free{\phi'}=V'\cap S=\var{\varphi_1}$ and recall that $(\var{\phi'}\backslash \free{\phi'})\cap \var{\psi}=\emptyset$. Hence, contracting $\calT'$ into a single node whose label is the constraint $\varphi_1$, results in a join tree of the formula  $\varphi_1\wedge \psi$. Thus this latter formula is acyclic.

We iterate this process with the $S$-components of the subformula $\psi$. When each $S$-component has been treated, $\phi$ is replaced by a quantifier free formula $\varphi=\varphi_1\wedge \ldots \wedge \varphi_m \wedge \phi_0$ where each $\varphi_i$ for $i=1, \ldots, m$ is atomic and $m$ is the number of $S$-components of $\phi$. Furthermore, $\phi_0$ is the conjunction of all atomic formulas of $\phi$ that contain only free variables. Also $\var{\varphi}=S$. Similarly $\calS$ is replaced by a structure $\calS'$ of size bounded by  $\|\Phi\|^{O(k)}$ (recall that each component is treated separately) such that  $\phi(\calS)=\varphi(\calS')$. In each iteration step the formula stays acyclic. Hence, $\varphi$ is acyclic and we conclude by applying Theorem~\ref{thm:constructVP}.
\end{proof}

\subsection{$\sW{1}$-hardness of parameterized $\sACQ$}\label{sct:starhard}

In this section we show that several parameterized versions of quantified $\sACQ$ are not fixed parameter tractable under standard assumptions from parametrized complexity. We consider the following parameterized counting problems:

\begin{itemize}
 \item $p$-$\mathrm{star}$-$\sACQ$: counting parameterized by the quantified star size,
 \item $p$-$\mathrm{var}$-$\sACQ$: counting parameterized by the number of free variables,
 \item $p$-$\sACQ$: counting parameterized by the size of the conjunctive formula.
\end{itemize}

Clearly, for every ACQ instance $\Phi$ with formula $\varphi$ we have that the quantified star size is at most $|\var{\varphi}| \le |\varphi|$. Thus we get from $p$-$\mathrm{var}$-$\sACQ$ and $p$-$\sACQ$ might be easier than $p$-$\mathrm{star}$-$\sACQ$. The next lemma states that -- unless there is a severe collapses in parameterized complexity -- all three problems are not fixed parameter tractable. This is in contrast to the decision version which is even in $\p$ for all three problems.

\begin{lemma}\label{lem:swhardness}
$p$-$\mathrm{star}$-$\sACQ$, $p$-$\mathrm{var}$-$\sACQ$ and $p$-$\sACQ$ are all $\sW{1}$-hard. 
\end{lemma}
\begin{proof}
We reduce $p$-$\mathrm{\#DirPath}$, i.e. counting of paths of length $k$, to $\sACQ$ on stars. With the $\sW{1}$-hardness of $p$-$\mathrm{\#DirPath}$ \cite{FlumGrohe04} the result will follow. The basic observation is that there are $|E|^k$ ordered choices of $k$ edges with repetitions. Thus it suffices to count the number of choices that are not paths to compute the number of $k$-paths in a graph. A choice $e_1, \ldots, e_k$ is not a $k$-path, if and only if it has one of the following defects:
\begin{enumerate}
\item It is not a walk, i.e. there is an $i$ such that $e_i$ has as end vertex not the start vertex of $e_{i+1}$, or
\item a vertex is visited twice.
\end{enumerate}
We will encode these properties into a $\sACQ$-instance of polynomial size whose hypergraph is a $k$-star.

So let $G=(V,E)$ be the input in which we are supposed to count $k$-paths. We construct a $\sACQ$-instance $\Phi$ which has the variables $y_1, \ldots, y_k, z$. The $y_i$ have the domain $V\times V$, while $z$ has the more complicated domain $\{0,1\}\times [k]\times V\times V\times [k]\times [k]\times V$. Observe that the domains have polynomial size and all constraints will be binary, so $\Phi$ has polynomial size in $n$ and $k$.

For each $i\in [k]$ we add a binary constraint $E_i$ in the variables $y_i, z$. Out of these we build the formula \[\varphi:= \exists z \bigwedge_{i=1}^k E_i(y_i, z).\]
The $y_i$-variable will choose arbitrary potential start and end points of an edge. The role of the $z$-variable is to guess one of the defects described above that prevents the chosen vertices from describing a path. We systematically describe tuple set for $E_i$. 
The first component of an assignment can only take $0$ or $1$ and encodes if $z$ guesses either a defect in the walk structure or a double variable.

\begin{itemize}
 \item $z$ may guess that the end vertex of $e_i$ is $v$ while the start vertex of $e_{i+1}$ is $u$ for $u\ne v$. It does so by taking the value $(0,i,v,u,j_1, j_2, a)$ for arbitrary $j_1, j_2, a$. If this guess is true, then $y_i$ must have chosen an edge that indeed does end in $v$. So we add the tuples $\{((b, v), (0,i,v,u,j_1, j_2, a))\mid a,b,v,u\in V, (b,v)\in E, j_1, j_2\in [k]\}$.
\item The second defect that $z$ may guess is that the edges do not form a path, because the end vertex of $e_{i-1}$ is $v$ while the start vertex of $e_{i}$ is $u$. This results in the tuples $\{((u, b), (0,i-1,v,u,j_1, j_2, a))\mid a,b,v,u\in V, (b,v)\in E, j_1, j_2\in [k]\}$.
\item If $z$ predicts a defect preventing a walk in some other place, $E_i$ does not have to check this, so we accept everything. The resulting tuples are $\{((u, v), (0,j,c,d,j_1, j_2, a))\mid a,c,d,v,u\in V, (u,v)\in E, j_1, j_2\in [k], j\notin \{i-1,i\}\}$. These cover all cases of the edges not being a path.
\item If $z$ guesses that the vertex $v$ occurs at two different places in the potential path, it does so by specifying edges $e_{j_1}, e_{j_2}$ with $j_1\le j_2$ such that the start vertex of $e_{j_1}$ and the end vertex of $e_{j_2}$ is $v$. It does so by taking a value $(1,\ell, a,b, j_1, j_2, v)$ for $\ell, a, b$ arbitrary. If $i\notin \{j_1, j_2\}$ then $E_i$ does not have to check for an effect and accepts if $y_i$ encodes an edge. Thus we add the tuples $\{((u,w), (1,\ell, a,b, j_1, j_2, v) \mid a,b,u,v,w\in V, (u,w)\in E, \ell, j_1, j_2 \in [k], i\notin \{j_1, j_2\}\}$.
\item If $z$ guesses a double occurence of $v$ and $i=j_1$ we accept only if that guess is correct and $j_2 \ge i$. So we add the tuples $\{((v,w), (1,\ell, a,b, i, j_2, v) \mid a,b,v,w\in V, (v,w)\in E, \ell, j_2 \in [k]\}$. If $i=j_2$ we add analogous tuples.
\end{itemize}

It is easy to see that $\Phi$ accepts assignments to the $y_i$ if and only if each $y_i$ gets the end points of an edge in $G$ and there is a defect that prevents the edges from being a path. Thus the number of satisfying assignments of $\Phi$ is the number of ordered choices of edges in $G$ with repetition that are not paths. This completes the proof.
 \end{proof}

\subsection{Bounded quantified star size is necessary}

In this section we show that quantified star size is in a sense the only restriction that makes $\sACQ$ tractable. Not only does bounded quantified star size give tractable instances, but the other way round under a standard assumption from parameterized complexity all classes of tractable $\sACQ$-instances must have bounded quantified star size. This is somewhat similar to the results of Grohe et al \cite{GSS-01} who proved that under reasonable assumptions the only polynomial time decidable subclass of bounded arity $CQ$ is the class of bounded treewidth.

As we have seen in the previous sections, not only the hypergraph of the input formula is decisive for tractability but also the structure of the quantified variables in this hypergraph. We formalize this in the following definition.

\begin{definition}An $S$-hypergraphs is a pair $(\calH, S)$ where $\calH=(V,E)$ is a hypergraph and $S\subseteq V$. We say that $\sACQ$ is tractable for a class $\calG$ of $S$-hypergraphs if for all $\sACQ$ instances $\Phi$ with the associated hypergraph $\calH$ of $\Phi$ and the set $S$ of free variables of $\Phi$ with $(\calH, S)\in \calG$ we can solve $\sACQ$ in polynomial time.
\end{definition}

\begin{example}
 Let $\calG$ be the class of acyclic $S$-hypergraphs of $S$-star size bounded by $k$. Then the result of Theorem \ref{thm:countingeasy} can be expressed as  ``$\sACQ$ is tractable for $\calG$''.
\end{example}

We will use the fact that $\sACQ$ is already hard for very restricted $S$-hypergraphs, namely for stars in which only the center is not in $S$. We call this class $\calG_S$. Observe that the proof of Lemma \ref{lem:swhardness} gives the following Lemma.

\begin{lemma}\label{lem:hardforstars}
 $\sACQ$ is $\sW{1}$-hard for $\calG_S$ parameterized by the size of the stars.
\end{lemma}

We now show the main result of this section.

\begin{theorem}\label{thm:starsizenecessary}
 Assume $\mathbf{FPT}\ne \sW{1}$, and let $\calG$ be a recursively enumerable class of acyclic $S$-hypergraphs. Then $\sACQ$ is polynomial time solvable for $\calG$ if and only if $\calG$ is of bounded $S$-star size.
\end{theorem}
\begin{proof}
 One direction of the claim is Theorem \ref{thm:starsizeeasy}. For the other direction assume that there is a class $\calG$ of unbounded $S$-star size such that $\sACQ$ is tractable on $\calG$. We show that in this case $\sACQ$ on $\calG_S$ parameterized by the star size is in $\mathbf{FPT}$ and with Lemma \ref{lem:hardforstars} we get $\mathbf{FPT}= \sW{1}$.

So all we have to do is to construct a  fixed parameter algorithm for $\sACQ$ on $\calG_S$. Let $\Phi$ be an instance of this problem, i.e. $\Phi$ has the formula $\varphi:= \exists z \bigwedge_{i=1}^k E_i(y_i, z)$. Let the domain of $\Phi$ be $D$. Because $\calG$ is recursively enumerable and of unbounded $S$-star size, there is a computable function $g:\mathbb{N}\rightarrow \mathbb{N}$ such that for $k\in \mathbb{N}$ one can compute $(\calH, S)\in \calG$ such that $\calH$ is of $S$-star size at least $k$ in time $g(k)$. We will embed $\Phi$ into $\calH$ to construct an $\sACQ$-instance $\Psi$ of size $g(k)n^{O(1)}$ where $n$ is the size of $\Phi$. Furthermore, $\Psi$ will have the $S$-hypergraph $\calH$ and the same domain $D$ as $\Phi$.

Let $\calH'=(V', E')$ be an $S$-$k$-star in $\calH$ that is formed by $Y=\{y_1, \ldots , y_k\}$, $Y\subseteq S$. For each edge $e\in E$ we define a constraint $E_e$. Let first $e\in E'$ be an edge that contains $y_i$ for some $i\in [k]$, then $E_e$ has as variables the vertices of $e$. Let $y_i$ be the first variable of $E_e$ followed by the other variables in $e\cap S$ and after those the variables in $e\setminus S$. Then $E_e$ has the tuples $\{(a, d, \ldots, d, b, \ldots, b)\mid (a,b)\in \calR_i\}$, where $\calR_i$ is the relation of $E_i$ and $d$ is an arbitrary but fixed value in $D$. Observe that this forces all variables in $(e\cap S)\setminus \{y_i\}$ to the variable $d$ in satisfying assignments, while the variables in $e\setminus S$ all have a common value $b$. Furthermore, observe that no two of the $y_i$ share an edge in $E'$, so $E_e$ is always well defined. 

Let $e\in E'$ with $e\cap Y = \emptyset$. Again we define a constraint $E_e$. Let in $E_e$ the first variables be those in $S\cap e$ followed by those in $e\setminus S$, then $E_e$ has the tuples $\{(d,\ldots, d, a, \ldots, a)\mid a\in D\}$ for the same $d\in D$ as before. Again in the satisfying assignments all variables in $e\cap S$ are forced to $d$, while the variables in $e\setminus S$ can take an arbitrary but equal value.

For $e\in E\setminus E'$ with $e\cap V'= \emptyset$, we add a constraint $E_e$ in with the single tuple $(d, \ldots, d)$. If $e\in E\setminus E'$ with $e\cap V' \ne \emptyset$, we have $e\cap V'\subseteq S$. Furthermore at most one vertex in $e\cap V'$, say $y_i$, can be in $Y$, because $Y$ forms a star. If there is no such $y_i$, we construct a constraint $E_e$ with the only tuple $(d, \ldots , d)$. If there is $y_i\in e$, we construct a constraint $E_e$ in which $y_i$ is the first variable with the tuples $\{(a, d, \ldots, d)\mid a \in D\}$. 

These are all constraints of $\Psi$. Let the formula 
\[\psi' := \bigwedge_{e\in E} E_e\]
 and $\psi$ the formula that we get from $\psi'$ by quantifying all variables in $V\setminus S$. Let $\psi$ be the formula of $\Psi$, then it is easy to see that $\Psi$ has has the associated $S$-hypergraph $(\calH,S)$. Furthermore, $\Psi$ has the same number of satisfying assignments as $\Phi$. This is because in each satisfying  assignment $\tu a$ of $\Psi$ all variables in $(V\setminus V')\cup (S\setminus Y)$ are set to $d$. Furthermore, all variables in $V'\setminus S$ take one common value $b$ in $\tu a$. Let $\tu a'$ be an assignment to $\Phi$ that we get by setting $\tu a'(z):=b$ and $\tu a'(y_i) := \tu a(y_i)$. It is easy to see that by this construction the satisfying assignments of $\Psi$ and $\Phi$ correspond directly, so the number of satisfying assignments is the same.

Now assume that $\sACQ$ is polynomial time solvable for $\calG$. It follows that the satisfying assignments of $\Psi$ can be counted in time $(g(k)n)^{O(1)}$ and thus $\sACQ$ on $\calG_S$ is in $\mathbf{FPT}$. With Lemma \ref{lem:hardforstars} this contradicts the assumption which completes the proof.
\end{proof}

\section{Applications to arithmetic circuit complexity}~\label{sec:arithmeticcircuitcomplexity}

We now show an adaptation of some of the results in this paper to arithmetic circuit complexity, i.e. the so-called Valiant model (\cite{Valiant-79}). A \emph{polynomial family} is a sequence $(f_n)$ of multivariate polynomials over a field $\mathbb{F}$.
The class of polynomial families of polynomial degree computed by families of polynomial size arithmetic circuits is denoted by $\VP$. This class  is a natural candidate to formalize efficient computation with arithmetic circuits.
A family $(f_n)$ of polynomials is in $\VNP$, if there is a family $(g_n)\in \VP$ and a polynomial $p$ such that $f_n(X) = \sum_{e\in \{0,1\}^{p(n)}} g_n(e,X)$ for all $n$ where $X$ denotes the vector $(X_1, \ldots, X_{q(n)})$ for some polynomial~$q$. By definition $\VP\subseteq \VNP$ but the precise relations between $\VP$ and $\VNP$ are still unknown.  It is however widely conjectured that $\VP\ne \VNP$, i.e. that not all polynomial familes in $\VP$ are efficiently computable by arithmetic circuits. 

A polynomial $f$ is called a \emph{projection} of $g$ (symbol: $f \le g$), if there are values $a_i \in \mathbb{F} \cup \{X_1, X_2, \ldots\}$ such that $f(X) = g(a_1, \ldots, a_q)$. A family $(f_n)$ of polynomials is a $p$-projection of $(g_n)$ (symbol: $(f_n) \le_p (g_n)$), if there is a polynomial $r$ such that $f_n \le g_{r(n)}$ for all $n$.
As usual we say that $(g_n)$ is hard for an arithmetic circuit class $\mathcal{C}$ if for every $(f_n) \in \mathcal{C}$ we have $(f_n) \le_p (g_n)$. If further $(g_n) \in \mathcal{C}$ we say that $(g_n)$ is $\mathcal{C}$-complete.

Classes of polynomials are often better understood through the natural polynomial families they contain or, better, which are complete for them. For example, it is well-known that determinant family, denoted $(\Det_n)$, of matrices  $(X_{i,j})_{i,j\in [n]}$ defined by: 
\[\Det_n(X_{i,j})= \sum_{\sigma\in S_n}\sgn(\sigma) \prod_{i=1}^n  X_{i,\sigma(i)}.\]

\noindent where $\sgn(\sigma)\in\{-1,1\}$ is the sign of the permutation $\sigma$, is contained in $\VP$ (although presumably not complete).  Similarly, the polynomial family ($\Perm_n$) representing the permanent of the matrices $(X_{i,j})_{i,j\in [n]}$ and defined as

\[\Perm_n(X_{i,j})= \sum_{\sigma\in S_n}\prod_{i=1}^n X_{i,\sigma(i)}.\]

\noindent has been shown to be $\VNP$-complete (see~\cite{Valiant-81}). Recently, the complexity of polynomial families defined by natural extensions of counting solutions to constraint satisfaction problems (like the the $Q$-polynomial of this paper) has been investigated (see \cite{BKM-11,Mengel-11} and also \cite{Bri-11}) and several non-circuit characterizations of $\VP$, $\VNP$ and other classes have been obtained. In this section, we generalize some of these results by showing that acyclic conjunctive queries can define polynomial families that characterizes $\VP$ and $\VNP$. 

\subsection{Tractable polynomials}

The weighted counting problem $\wsACQ{\mathbb{F}}$ can be seen as the problem of computing the polynomial
\[Q(\Phi) := \sum_{a\in \phi(\calS)} \prod_{x \in \var{\phi}} X_{a(x)},\] in the variables $\{X_d \mid d \in D\}$
for a  $\Phi$ in $\CQ$. This naturally makes $\wsACQ{\mathbb{F}}$ or equivalently computing $Q(\Phi)$ a question in the Valiant model. We have:

\begin{theorem} \label{thm:VPbyACQfree}
If $(\Phi_n)$ is a family of 
%acyclic conjunctive queries 
$\ACQ$ of polynomially bounded size and bounded quantified star size, then $(Q(\Phi_n)) \in\VP$. Moreover, any family in $\VP$ is a p-projection of $(Q(\Phi_n))$, where the $\Phi_n$ are polynomial size quantifier free conjunctive queries whose hypergraph is a~tree.\end{theorem}
\begin{proof}
The upper bound follows from the proof of Theorem~\ref{thm:countingeasy} and Theorem~\ref{thm:constructVP}. The lower bound is already true for acyclic queries on graphs and follows from~\cite{Mengel-11}.
\end{proof}

%This gives a non-circuit characterization of $\VP$ (see also \cite{BKM-11,Mengel-11,Bri-11}). Also 
This result  shows that while $\wsACQ{\mathbb{F}}$ is a tractable counting problem, it is probably harder than computing the determinant, which is quite rare in counting complexity (see \cite{MP-08} for the role of the determinant in the Valiant model).

As a corollary, we also obtain the following result from~\cite{Mengel-11}.

\begin{corollary} Let $k,d$ be integers. For any family $(\Phi_n)$ of $\CSP$ of polynomially bounded size and arity bounded by $d$ and tree-width bounded by $k$, then  $(Q(\Phi_n)) \in\VP$.
\end{corollary} 

\begin{proof} Comes from the fact that any $\CSP$ of tree-width $k$ built on relation of arity $d$ can be transformed (by taking joins of atoms in each bag of the tree decompositions) in time $O(n^{f(k,d)})$, for some function $f$, into an acyclic $\CSP$  with the same set of solutions.  
\end{proof}

\subsection{The power of existentially quantified variables}\label{sct:quantifiedchanged}

In this section we show a version of Pichler and Skritek's hardness result for $\sACQ$ for the Valiant model, i.e. the polynomial  $Q(\Phi)$ appears to be harder to compute than in the unquantified case. We state the upper bound in a more general way:

Let $D$ be a finite set with $|D|= d$. For a positive integer $n$ we encode an elements of $\tu a = (a_1,\ldots, a_n)\in D^n$ by a $d\times n$-matrix $M= (m_{i,j})_{i\in D, j\in [n]}$ such that $m_{i,j}= 1$ if $a_j = i$ and $m_{i,j} = 0$ otherwise. Observe that a $0$-$1$-matrix encodes an element in $D^n$ if and only if in each column there is exactly one $1$. We define the monomial $q(M) :=q(a) = \prod_{j\in [n]}X_{a_j}$ if $M$ encodes $a$ and $q(M):= 0$ otherwise. With this notation we can prove a version of Valiant's classic criterion \cite{Valiant-79} for the $Q$-polynomial.

\begin{proposition}\label{prop:valiantcriterion}
 Let $\alpha:\{0,1\}^* \rightarrow \mathbb{Z}$ a function in $\gapP$. Furthermore let $(D_n)$ be a polynomially bounded family of sets and let $p(n)$ be a polynomial. Then the family $(f_n)$ defined by 
\[f_n= \sum_{M\in \{0,1\}^{|D_n|\times p(n)}} \alpha(M) q(M)\]
is in $\VNP$.
\end{proposition}
\begin{proof}
 It is folklore that $\alpha$ can computed as $\alpha(x) = \sum_{e\in \{0,1\}^{r(n)}} g_{|x|}(x,e)$ where $r$ is a polynomial and $g_n$ is a family of uniform arithmetic formulas of polynomial size (see e.g. \cite{BF91}). In the proof of Lemma 4 in the full version of \cite{Mengel-11} it is shown how to compute the function $q$ with small arithmetic formulas. Combining this we directly get the proposition.
\end{proof}

\begin{proposition} \label{prop:quantifiedhard}
If $(\Phi_n)$ is a family of acyclic conjunctive queries of polynomial size, then $(Q(\Phi_n)) \in\VNP$. Moreover, any family in $\VNP$ is a p-projection of such a $(Q(\Phi_n))$. The $(\Phi_n)$ family can be supposed of arity bounded by two.
\end{proposition}
\begin{proof}
We start off with the containment in $\VNP$, which follows easily from the fact that given an assignment $a$ and an acyclic  conjunctive query $\Phi$ one can in polynomial time decide if the $a$ satisfies $\Phi$. Applying Proposition \ref{prop:valiantcriterion} we get the upper bound.

The hardness is obtained by reduction from the family ($\Perm_n$).
$\Perm_n$ can be seen as the sum of weights of perfect matchings in the weighted bipartite graph $K_{n,n}$.

Let $K_{n,n}=(A\cup B,E)$ be the complete bipartite graph with $A=\{a_1,...,a_n\}$,  $B=\{b_1,...,b_n\}$. We denote by $e_{i,j}$ the edge between vertices $a_i\in A$ and $b_j\in B$. We construct a structure $S=\langle \calU, F, G, H, I, J, K, a_1,...,a_n\rangle$ of domain $\calU$ as follows.

\begin{itemize}
\item The universe $\calU= A \cup B \cup \{e_{i,j}\mid i,j\in [n]\} \cup \{p,n, l\}$,
\item $G= \{p,n\}$,
\item $F= \{(n,l)\}\cup \{(p, e_{i,j})\mid i,j\in [n]\}$,
\item $H= \{(e_{i,j}, a_i)\mid i,j\in [n]\}\cup \{(l,a_i)\mid i\in [n]\}$,
\item $I= \{(p,l)\} \cup \{(n, b_i) \mid i\in [n]\}$,
\item $K= \{(b_j,e_{i,k}): i,j,k\in [n], j\ne k\} \cup \{(l,l)\}$
\end{itemize}
 
 Note that the maximal arity of a predicate is two.
 Let $\phi(x_1,...,x_n,x,x_1',...,x_n')$ be the following acyclic conjunctive query:
 
% \[
 %\phi:=
 \begin{eqnarray}~\label{eqn:formulahardnew}
 G(x) \wedge \bigwedge_{i=1}^n F(x, x_i) \wedge \bigwedge_{i=1}^n H(x_i,a_i) \wedge \exists y I(x,y) \wedge \bigwedge_{i=1}^n K(y, x_i') \wedge \bigwedge_{i=1}^n H(x_i',a_i)
 \end{eqnarray}
%\]

There are two types of satisfying assignments:

\begin{itemize}
 \item If $x$ takes the value $p$, then $y$ and all $x_i'$ must take the value $l$. The $x_i$ take as values the edges $e_{i,j}$ in such a way that for each $i\in [n]$ there is an edges $e_{i,j}$. Thus the vertices in $A$ are mapped to the vertices of $B$ in the original graph in an arbitrary way.
\item If $x$ takes the value $n$, then all $x_i$ must take the value $l$. Furthermore $y$ takes a value $b\in B$. The $x_i'$ then take as values the edges $e_{i,j}$ in such a way that each $a_i$ is mapped to a vertex $b_j\in B\setminus \{b\}$ by this edge. Thus the assignment to the $x_i'$ is an arbitrary non-injective assignment of the vertices in $A$ to those in $B$.
\end{itemize}

Thus the query $\Phi=(\calS,\phi)$ defines a polynomial $Q(\Phi)$ with the following property (with $e_{i,j}$ corresponding to variable $X_{i,j}$, $p$ to $X_p$, $n$ to $X_n$ and $l$ to $X_l$).
 
\[ 
 \begin{array}{rl}
Q(\Phi) & =  \displaystyle \sum_{\phi(S)} \prod_{t\in \var{\Phi}} X_{a(t)} \\
& = \displaystyle X_p X_l^n \sum_{\sigma\in [n]^{[n]}}\prod_{i=1}^n X_{i,\sigma(i)} + X_n X_l^n \sum_{\sigma\in [n]^{[n]}\backslash S_n}\prod_{i=1}^n X_{i,\sigma(i)}.
\end{array}
\]

Projecting correctly, we get

 \[\Perm_n(X_{i,j})=Q(\Phi)(X_1,...,X_n,X_p, X_n, X_l)|_{X_{p} = 1, X_n= -1, X_l=1}.\]
%\arnaud{Actually, if correct, the proof gives more that simple p-projection, no ? The query $\Phi$ is constructible in polynomial time from the original graph and this should be said in the statement.}\stefan{Well, this depends on if the reduction to the permanent is somehow uniform. This is true if for example the arithmetic circuits computing $(f_n)$ are uniform.}
\end{proof}

\subsection{Unions and intersections of ACQ}

We now show that a version of Proposition \ref{prop:union} is also true for the Valiant model.

\begin{proposition}
If $(\Phi_n)$ is a family of queries of polynomial size that are conjunction (resp. disjunction) of two acyclic conjunctive queries, then $(Q(\Phi_n)) \in\VNP$. Moreover, any family in $\VNP$ is a p-projection of such a $(Q(\Phi_n))$. The result remains true for the polynomial family $(P(\Phi_n))$.
\end{proposition}

\begin{proof}[Proof (Sketch)]
The upper bound follows directly from Proposition \ref{prop:valiantcriterion}.
The proof of the lower bound for conjunction of acyclic queries follows directly as Proposition~\ref{prop:union}. The case of disjunction is obtained by reduction from the case of conjunction. Let $\Phi=(\calS,\phi(\tu x))$ and $\Psi=(\calS,\psi(\tu x))$ be two acyclic conjunctive queries. W.l.o.g. we can suppose they both are on the same structure $\calS$ of signature $\sg$ and domain $D$. We denote by $\Phi\wedge\Psi$ the instance $(\calS, \phi(\tu x)\wedge \psi(\tu x))$ and by $\Phi\vee\Psi$ the instance $(\calS, \phi(\tu x)\vee \psi(\tu x))$. 
Let $\calS'$ be a new structure of domain $D'=D \cup \{\alpha_1,\alpha_2, \alpha_3\}$ where $\alpha_1,\alpha_2, \alpha_3$ are not in $D$. Structure $\calS'$ includes $\calS$ and is equipped with two new unary relations $R$ and $S$ which are defined as follows:

\[R=\{\alpha_1,\alpha_2\}, S=\{\alpha_2,\alpha_3\}.\]

Let us now consider the following disjunction of two acyclic formulas:

\[\varphi(\tu x,y) \equiv \left(\phi(\tu x)\wedge R(y)\right) \vee \left(\psi(\tu x)\wedge S(y)\right). \]

The query problem $\Upsilon=(\calS',\varphi(\tu x,y))$ has the following tuples as solutions:

\begin{itemize}
\item $(\tu a,\alpha_1)$ for $\tu a \in \phi(\calS)$.
\item $(\tu a,\alpha_3)$ for $\tu a \in \psi(\calS)$.
\item $(\tu a,\alpha_2)$ for $\tu a \in \phi(\calS)\cup \psi(\calS)$.
\end{itemize}

Then, associating each value $\alpha_i$ with variable $Y_i$:

\[ 
 \begin{array}{rl}
Q(\Upsilon) & =  \displaystyle Y_1Q(\Phi) + Y_3 Q(\Psi) + Y_2Q(\Phi\vee\Psi).
\end{array}
\]

By projection, we get

 \[Q(\Phi\wedge\Psi)=Q(\Upsilon)(\tu X,Y_1,Y_2,Y_3)|_{Y_{1} = 1, Y_2=-1, Y_3=1}.\]

This shows that polynomials obtained by disjunction of two acyclic queries can represented as projections of polynomials obtained by conjunction and hence this is true for all polynomial families in $\VNP$. 
\end{proof}

\section{Conclusion}

We have presented a complete picture of tractability for weighted $\sACQ$. However, there are still many open questions that could be explored in the future.

The first question is how to generalize our results from acyclic to more general classes of conjunctive queries. While it should be possible to generalize the unquantified case to, say, bounded hypertree width (and indeed Pichler and Skritek \cite{PS-11} sketch this for unweighted counting), the quantified case is less clear. The generalizations of quantified star size for e.g. bounded treewidth queries and thus the counting algorithm of Theorem \ref{thm:countingeasy} appears straightforward. The real problem would then be if there is an efficient algorithm to decide for a combination of quantified star size and treewidth similarly to Theorem \ref{thm:starsizeeasy}. For a fixed tree decomposition computing the star size would be easy, but ruling out other tree decompositions with smaller quantified star size is at least non obvious to us.

One future direction of work could be trying to apply our results to combinatorial counting problems. Conjunctive queries are a very versatile in encoding other problems, so can our results be used to find non-obvious algorithms for such problems?

We now turn to aspects from parameterized complexity. The first question is completely determining the complexity of the parameterized problems in Section \ref{sct:starhard}. For the $p$-$\sACQ$ it is easy to see that (by complementing the formula) the problem can be reduced to counting solutions of $\Pi_1$-formulas. These formulas define the class $\mathbf{\#A[2]}$ and thus we have $p$-$\sACQ\in \mathbf{\#A[2]}$. Observe that we do not use the acyclicity in this reduction at all which could make the problem easier (and does indeed in other settings). So it would not be very surprising if we could show $p$-$\sACQ\in \mathbf{\#W[1]}$. For $p$-$\mathrm{var}$-$\sACQ$ containment in $\sW{P}$ is straightforward, but we do not see how to bring this down into the $\mathbf{\#A}$- or the $\mathbf{\#W}$-hierarchy (see \cite{FlumGrohe06} and also the discussion in \cite{Thurley-06} for a definition of these classes). It would not even be surprising to show that $p$-$\mathrm{var}$-$\sACQ$ could be reduced to $p$-$\sACQ$. Intuitively, too many quantified variables for few free variables should not make too much sense, which could then lead directly to a bound of the size of formulas in the number of free variables. Finally, for $p$-$\mathrm{star}$-$\sACQ$ not even an $\sW{P}$ upper bound is apparent.

Our Theorem \ref{thm:starsizenecessary} strongly depends on the fact that we can in the construction of hard instances use as many different relation symbols $E_e$ as we need. Can we show a stronger version in parallel to the result of Grohe et al. \cite{GSS-01} that holds for any fixed vocabulary that contains at least one binary relation symbol?

Another question is if there are any structural subclasses of $\sACQ$ that allow more efficient counting than the algorithm of Theorem \ref{thm:countingeasy}. Theorem \ref{thm:starsizenecessary} tells us that these classes would be of bounded quantified star size, but could we add more structural restrictions to yield a class that allows, say, fixed parameter counting? Or can we under plausible assumptions rule out such classes in the style of \cite{Marx-07}?

\bibliographystyle{alpha}
\bibliography{bibfile}

\end{document}